\documentclass[11pt]{article}

\usepackage{hyperref}

\usepackage{amsmath}
\usepackage{amsthm}
\usepackage{amssymb}
\usepackage{graphicx}
\usepackage{color}
\usepackage{tikz}
\usepackage{wrapfig}
\usepackage{mathtools}
\usepackage{bbm}
\usepackage{graphicx} 
\usepackage{caption,subcaption}
\usepackage{adjustbox} 
\usepackage{rotating}
\usepackage{mathrsfs}

\usepackage{algorithmic}

\usepackage{dsfont} 

\usepackage[margin=1in]{geometry}
\usepackage{setspace}

\newcommand{\indep}{\mathbin{\rotatebox[origin=c]{90}{$\models$}}}

\newcommand{\E}{\mathbb{E}}

\newcommand{\indicator}[1]{\mathds{1}\{ #1 \}}

\newcommand{\X}{\mathbf{X}}

\newcommand{\C}{\mathbf{C}}

\newcommand{\x}{\mathbf{x}}

\newcommand{\bL}{\mathbf{L}}
\newcommand{\bl}{\mathbf{l}}

\DeclareMathOperator{\DE}{DE}
\DeclareMathOperator{\IDE}{IDE}
\DeclareMathOperator{\AR}{AR}
\DeclareMathOperator{\VE}{VE}
\DeclareMathOperator{\CVE}{CVE}
\DeclareMathOperator{\CE}{CE}
\DeclareMathOperator{\SE}{SE}
\DeclareMathOperator{\IE}{IE}
\DeclareMathOperator{\SAR}{SAR}
\DeclareMathOperator{\I}{I}
\DeclareMathOperator{\net}{net}

\newcommand\redsout{\bgroup\markoverwith{\textcolor{red}{\rule[0.5ex]{2pt}{0.4pt}}}\ULon}

\newtheorem{thm}{Theorem}
\newtheorem{lem}{Lemma}
\newtheorem{defn}{Definition}
\newtheorem*{defn*}{Definition}

\newtheorem{cor}{Corollary}

\newtheorem{assumption}{Assumption}

\usepackage[numbers,sort&compress]{natbib}

\title{Identification of causal intervention effects under contagion}

\author{Xiaoxuan Cai$^1$, Wen Wei Loh$^2$, and Forrest W. Crawford$^{1,3,4,5}$ \\[1em]
 \normalsize 1. Department of Biostatistics, Yale School of Public Health \\
 \normalsize 2. Department of Data Analysis, University of Ghent \\
 \normalsize 3. Department of Statistics \& Data Science, Yale University \\
 \normalsize 4. Department of Ecology and Evolutionary Biology, Yale University \\
 \normalsize 5. Yale School of Management}

\begin{document}
\maketitle

\begin{abstract} 
\noindent Defining and identifying causal intervention effects for transmissible infectious disease outcomes is challenging because a treatment -- such as a vaccine -- given to one individual may affect the infection outcomes of others.  Epidemiologists have proposed causal estimands to quantify effects of interventions under contagion using a two-person partnership model.  These simple conceptual models have helped researchers develop causal estimands relevant to clinical evaluation of vaccine effects. However, many of these partnership models are formulated under structural assumptions that preclude realistic infectious disease transmission dynamics, limiting their conceptual usefulness in defining and identifying causal treatment effects in empirical intervention trials.  In this paper, we propose causal intervention effects in two-person partnerships under arbitrary infectious disease transmission dynamics, and give nonparametric identification results showing how effects can be estimated in empirical trials using time-to-infection or binary outcome data. The key insight is that contagion is a causal phenomenon that induces conditional independencies on infection outcomes that can be exploited for the identification of clinically meaningful causal estimands. These new estimands are compared to existing quantities, and results are illustrated using a realistic simulation of an HIV vaccine trial.  \\[1em]
\end{abstract}


\section{Introduction}

Estimating the causal effect of an intervention can be challenging when the outcome of interest is contagious \citep{ogburn2018challenges}.  For example, a vaccine intended to prevent infection by a transmissible disease may reduce the risk of infection in individuals who receive it, and may reduce transmissibility if a vaccinated individual becomes infected.  When study subjects are embedded in interacting groups among whom the disease may be transmitted, it can be difficult to separate the effect of one subject's vaccination on themselves from its effect on other individuals and the group as a whole.  Usually, the estimand of greatest clinical interest is the effect of an intervention on individual risks of infection, holding all else constant. 

The pursuit of empirically meaningful definitions of population-level causal vaccine effects has a long history. \citet{greenwood1915statistics} first described informally the conditions under which vaccine effects can be estimated. \citet{halloran1991direct} established some of the first theory and definitions for clinically meaningful vaccine effects, and subsequent work by 
Halloran and colleagues \citep{halloran1991study, halloran1997study, halloran1999design} described epidemiological study designs for identifying these quantities. \citet{halloran1995causal} gave the first formal definitions of causal vaccine estimands using notation and assumptions of a modern counterfactual-based causal inference framework \citep{rubin2005causal}. \citet{hudgens2008toward} and \citet{perez2014assessing} showed how this formalism could be applied in empirical randomized trials of clustered individuals \citep{halloran2010design,halloran2016dependent}.  More recently, researchers have shown that randomized trials may not measure clinically meaningful intervention effects when infection can be transmitted within groups \citep{morozova2018risk,eck2019randomization}.  

Researchers have described two-person partnership models of infectious disease transmission for defining more granular, or individual, causal intervention effects.  \citet{vanderweele2011bounding} introduced a partnership model consisting of two interacting individuals who may be vaccinated and can transmit the infection to each other.  By limiting the extent of potential disease transmission to two individuals, effects can be more easily defined in terms of
potential outcomes indexed by treatments of both individuals and the outcome of their partner. Using a principal stratification approach, the partnership
model permits computations of bounds for the infectiousness effect \citep{vanderweele2011bounding,chiba2012note,halloran2012causal}.

\citet{vanderweele2012components} presented a special case of the partnership model in which one partner is home-bound, and can only be infected from within via infection of the other. The assumed asymmetry in the disease transmission structure makes this model tractable for point identification of contagion and infectiousness effects by ensuring that interference only happens in one direction.
Interference arises when an individual's potential outcomes depend on the treatment status of others \citep{Cox:1958aa}. 
To allow for mutual dependence of potential outcomes on each other's treatments, \citet{ogburn2014causal,ogburn2017vaccines} extend this approach using causal diagrams and mediation arguments for symmetric pairs. More recent work on symmetric mediation provides new tools for accommodating structural assumptions about the nature of dependence in outcomes under different forms of interference \citep{shpitser2017modeling,sherman2018identification,ogburn2018causal,bhattacharya2019causal}.  

Statisticians and epidemiologists have developed parallel literature devoted to mathematical modeling of infectious disease transmission dynamics.  This work treats infectious disease transmission as a dynamic temporal phenomenon: the risk of infection in a given subject may change over time, as a function of the infection status of their contacts, and covariates.  For example, \citet{rhodes1996counting} present hazard models of infectious disease transmission in groups that accommodate individual-level (e.g. treatment) variables with possibly different effects on susceptibility and infectiousness.  \citet{kenah2013non,kenah2015semi} extends these ideas to develop nonparametric and semi-parametric statistical models for estimating covariate effects under contagion.  Structural transmission modeling has gained wide use in clinical studies of infectious disease dynamics because it combines mechanistic assumptions about infectious disease transmission with regression-style covariate adjustment \citep{auranen2000transmission,oneill2000analyses,becker2003estimating,cauchemez2004bayesian,cauchemez2006investigating,cauchemez2009household,becker2006estimating,yang2006design,tsang2015influenza,tsang2016individual}. 

In this paper, we take a different approach to defining and identifying causal identification for intervention effects in general symmetric two-person partnerships under contagion by formalizing the role of time in infectious disease transmission from a causal perspective.  In our construction, either individual can be vaccinated, can be infected from outside, and can infect the other if infected themselves. Treatments and individual covariates may affect both susceptibility to, and infectiousness of, the infection outcome. We first introduce a generic causal model and straightforward assumptions that permit identification of time-controlled and marginal contagion, susceptibility, and infectiousness effects.  The framework is nonparametric and imposes no restrictions on the joint distribution of infection times in a partnership. Our method generalizes those that treat infection outcomes as simultaneous mediating variables \citep{vanderweele2011bounding,chiba2012note,halloran2012causal}.  Before any infections have occurred in a partnership, the potential first infection times are conditionally independent, because neither partner can yet transmit the infection to the other.  
After the first infection, the causal structure of the system changes, and the time to infection of the remaining susceptible partner is now a function their partner's, as well as their own, treatment and covariates. Because the resulting causal model incorporates this temporally changing structure, it is more complex than those presented by other authors. On the other hand, this added complexity yields straightforward point identification results that cannot be obtained by treating the infection outcomes of both individuals as simultaneous mediating variables.  We discuss nonparametric identification under randomization and in observational settings, compare these estimands to existing quantities proposed by other authors, and conduct a simulation analysis of a hypothetical HIV vaccine trial to illustrate the estimands.


\section{Setting}

Consider a population consisting of pairs of individuals, henceforth referred to as partnerships. Within a partnership, either individual can be infected from an external source
(exogenous to the partnership), and once infected, an individual may internally (endogenous to the partnership) transmit the infection to their uninfected partner. In line with existing partnership models, it is assumed throughout that transmission between partnerships does not occur.  Label the individuals in the partnership 1 and 2. In a given partnership, let $T_i$ be the infection time of person $i$, and let $Y_i(t) = \indicator{T_i<t}$ be the indicator of prior infection at time $t$.  Let $X_i$ be the binary treatment status of $i$, and let $\X=(X_1,X_2)$ be the joint binary treatment vector for the partnership. Let $\bL=(\mathbf{L}_1,\mathbf{L}_2)$ be measured baseline covariates for the two individuals, including covariates for the partnership as a whole.  In each partnership, we observe $(T_1,T_2,X_1,X_2,\bL_1,\bL_2)$.  In a symmetric partnership, labels 1 and 2 are arbitrary and interchangeable. We will use the index $i$ to refer generically to one individual, either 1 or 2, and $j$ to refer to the partner of $i$. 

To describe the causal structure of infectious disease transmission within a partnership, we consider a decomposition of the infection time $T_i$ that will help us define counterfactual infection times under different circumstances.  Recall that both individuals are uninfected at baseline, and let $W_i$ be the time to initial infection of $i$ from a source of infection external to the partnership.  If $i$ is the first in their partnership to become infected, then we observe $W_i$.  If their partner $j$ is infected first, we observe $W_j=w_j$ and $W_i$ is \emph{censored} at time $w_j$.  When $W_i$ is censored by earlier infection of $j$, let $Z_i$ be the additional time to infection of $i$, beyond $w_j$. Formally, we decompose $T_i$ as follows.
\begin{equation}
  T_i =   \left\{ \begin{array}{ll}
               W_i & \textrm{if }W_i < W_j \\
               W_j + Z_i & \text{otherwise}.
           \end{array} \right.
 \label{eq:Tdef}
\end{equation}
We emphasize that the decomposition \eqref{eq:Tdef} is purely notational, and places no \emph{a priori} restrictions on the joint distribution of infection times $(T_i,T_j)$.  Rather, \eqref{eq:Tdef} permits specification of causal assumptions, outlined below, to capture the way treatments to both $i$ and $j$ may affect different parts of the waiting times to infection.  Figure \ref{fig:illustration} illustrates this decomposition and motivates the contagion effect presented in Definition \ref{defn:contagion} below: the disease is said to be ``contagious'' if the distribution of $T_i$ is different from that of $W_i$, or equivalently, if prior infection of $j$ changes the conditional distribution of the remaining time to infection of $i$.

\begin{figure}
\centering
\includegraphics{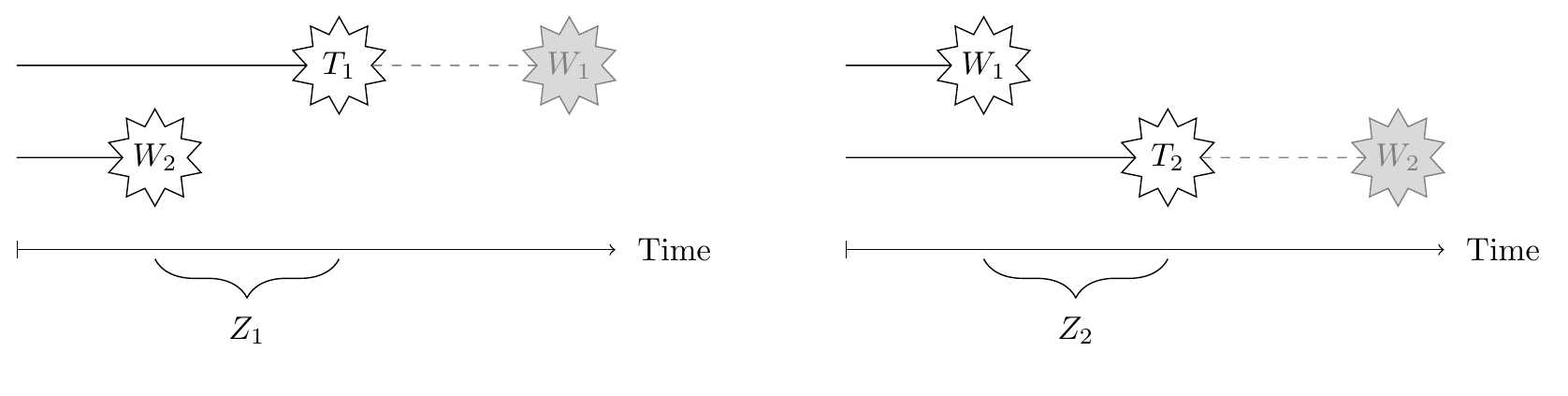}
\caption{Illustration of contagion in a two-person partnership. At left, when subject 2 becomes infected first ($W_2<W_1$), then $W_1$ is censored, and $Z_1$ is the remaining time to infection of subject 1. At right, when subject 1 is infected first ($W_1<W_2$), then $W_2$ is censored, and $Z_2$ is the remaining time to infection of subject 2.  Informally, the outcome is said to be ``contagious'' when the distribution of $T_i$ is different from that of $W_i$.} 
  \label{fig:illustration}
\end{figure}


\subsection{Assumptions}

In this section, we describe assumptions that are sufficient to identify scientifically meaningful and relevant causal effects, which are described later, from observable infection time data for each partnership.
We state assumptions for a generic individual $i$ and their partner $j$.  To define potential, or counterfactual, infection times for individual $i$, let $W_i(\x)$ be the potential value of $W_i$ under the joint treatment allocation $\x=(x_1,x_2)$. Let $Z_i(w_j,\x)$ be the additional potential time to infection of $i$, following the infection of $j$ at time $W_j(\x)=w_j$, under joint treatment allocation $\x$. 
\begin{assumption}[Exclusion restriction and independence of the initial infection]
$W_i(\x) = W_i(x_i)$,
$W_i(x_i) \indep W_j(x_j) \mid \bL$, and  
$W_i(x_i) \indep \bL_j \mid \bL_i$, for all $\x$.
\label{as:exclusion}
\end{assumption}
Assumption \ref{as:exclusion} states that individual $i$'s initial infection time $W_i(\x)$ is invariant to the partner's treatment status $x_j$, and may be viewed as a ``no-interference'' assumption on $W_i$, since $W_i$ is the initial infection time from an external source, which can only be realized when $W_i$ precedes $W_j$.
Further, $W_i(x_i)$ is independent of $W_j(x_j)$ given covariates $\bL$. 
\begin{assumption}[Initial infection exchangeability]
$Z_i(w_j,\mathbf{x}) \indep W_j(x_j) \mid \bL$, for all $\x$, $w_j>0$.
\label{as:infignorability}
\end{assumption}
Assumption \ref{as:infignorability} states that there is sufficient covariate information in $\bL$ so that the potential further time to infection $Z_i(w_j,\x)$ when $j$ is infected at $w_j$ is conditionally independent of the potential initial infection time $W_j(x_j)$ of $j$.
While this assumption bears similarity to the assumption of no unobserved confounding between the mediator and outcome counterfactuals (for the same individual) under the single mediator setting \citep{robins1992identifiability,pearl2001direct},
we note that this assumption relates to counterfactuals for different individuals.
\begin{assumption}[Treatment exchangeability] 
$W_i(x_i) \indep \X \mid \bL$ and $Z_i(w_j,\x) \indep \X \mid \bL$, for all $\x$, $w_j>0$.
\label{as:txignorability}
\end{assumption}
Assumption \ref{as:txignorability} means that the potential infection times $W_i(x_i)$ and $Z_i(w_j,\x)$ are independent of the assigned treatment $\X$ given covariates. This assumption resembles the conventional criterion of ``no unobserved confounding'' between treatment and outcome; in this context, Assumption \ref{as:txignorability} states that there is no unmeasured confounding between two infection times and the joint treatment vector in the same partnership.

Two additional assumptions ensure identifiability of potential infection outcomes from observational data. 
\begin{assumption}[Consistency]
$W_i= W_i(x_i)$, and $Z_i=Z_i(w_j,\x)$ under the observed treatment $\X=\x$ and infection time $W_j=w_j$, for all $\x$, $w_j>0$.
\label{as:consistency}
\end{assumption}
\begin{assumption}[Positivity]
$0<\Pr(W_j<w|X_i=x_i,\bL_i=\bl_i)<1$ for all $w>0$, $x_i$, and $\bl_i$; 
$0<\Pr(Z_j<z|\X=\x_i,\bL=\bl)<1$ for all $z>0$, $\x$ and $\bl$; 
and
$0<\Pr(\X=\x|\bL=\bl)<1$ for all $\bl$. 
\label{as:positivity}
\end{assumption}

A final assumption permits identification of certain ``natural'' or ``cross-world'' potential infection outcomes. 
\begin{assumption}[Cross-world initial infection exchangeability]
  $Z_i(w_j,\x) \indep W_j(x_j') \mid \bL$ when $\x=(x_i,x_j)$ and $x_j' \neq x_j$, for all $\x$, $x'_j$.
  \label{as:crossworld}
\end{assumption}

Finally, let $T_i(W_j(x_j),\x)$ be the potential outcome for the infection time of subject $i$, when $j$ is infected at time $W_j(x_j)$ and the assigned treatments are $\x=(x_i,x_j)$. Following the decomposition \eqref{eq:Tdef} and by Assumptions \ref{as:exclusion}--\ref{as:txignorability}, we can construct the potential infection time $T_i(W_j(x_j),\x)$ as follows:
\begin{equation}
  T_i(W_j(x_j),\x) =   \left\{ \begin{array}{ll}
               W_i(x_i) & \textrm{if } W_i(x_i) < W_j(x_j) \\
               W_j(x_j) + Z_i(W_j(x_j),\mathbf{x}) & \text{otherwise}.
           \end{array} \right. 
 \label{eq:Tdef2}
\end{equation}
The potential infection time with $W_j=w_j$ fixed is denoted as $T_i(w_j,\x)$.  For convenience, define the binary potential infection outcome evaluated at time $t$, $Y_i(t;w_j,\x)=\indicator{T_i(w_j,\x)<t}$ and $Y_i(t;W_j(x_j'),\x) = \indicator{T_i(W_j(x_j'),\x)<t}$. 

The potential infection time decomposition \eqref{eq:Tdef2} formalizes intuition about the structure of interference under contagion: there can be no interference without prior infection.  When neither $i$ nor $j$ is infected, the time to infection of $i$ is solely a function of the treatment $x_i$, and there is no interference within the partnership. This is because the treatment $x_j$ of $j$ can only affect $i$ \emph{after} $j$ becomes infected.  When $j$ is the first to be infected, the remaining time to infection of $i$ is now a function of both $x_i$ and $x_j$, because $j$ has now gained the ability to transmit to $i$.  This apparent complexity simplifies identification of causal effects, as we show below.s


\section{Identification of potential infection outcomes}

We wish to non-parametrically identify the average potential infection outcome $\E[Y_i(t;w_j,\x)]$ using observations of pairwise outcomes, treatments, and covariates $(T_i,T_j,X_i,X_j,\bL_i,\bL_j)$. A preliminary result identifies the distribution of $W_i(x_i)$ in Lemma~\ref{lem:Fi} using information about infection times. 

\begin{lem}
Suppose Assumptions \ref{as:exclusion}-\ref{as:positivity} hold. Then the distribution function of $W_i(x_i)$ given $\bL_i=\bl_i$ is identified by 
\[ F_i(w|x_i,\mathbf{l}_i) =  1-\exp\left[-\int_0^{w} \frac{p(T_i=u,T_j>u|\X=(x_i,x_j),\mathbf{L}=(\mathbf{l}_i,\mathbf{l}_j))}{\Pr(T_i>u,T_j>u|\X=(x_i,x_j),\mathbf{L}=(\mathbf{l}_i,\mathbf{l}_j))} du\right]  \]
for any fixed values of $x_j$, and $\mathbf{l}_j$, where $p(T_i=u,T_j>u|\X=(x_i,x_j),\mathbf{L}=(\mathbf{l}_i,\mathbf{l}_j))$ is the joint probability density of $T_i$ and survivor function of $T_j$. 
\label{lem:Fi}
\end{lem}

Lemma \ref{lem:Fi} is a standard distributional identification result in competing risks \citep{akritas2004nonparametric}; in this paper, $W_i$ and $W_j$ are competing event times within the same partnership. The identified distribution function $F_i(w|x_i,\bl_i)$ is a function of $x_i$ and $\bl_i$ only, and is invariant to values of $x_j$ and $\bl_j$. However, in order to identify this function, particular values of $x_j$ and $\bl_j$ must be held constant.

The main result shows that average exposure-controlled potential infection outcomes given $\bL=\bl$ are identified. Proofs are given in the Appendix.  
\begin{thm}[Identification of the average exposure-controlled potential infection outcome]
Suppose Assumptions  \ref{as:exclusion}-\ref{as:positivity} hold and $\x=(x_1,x_2)$. For fixed values of $w_j$ and $t$, if $w_j < t$, 
\begin{equation}
  \E[Y_i(t;w_j,\x)|\bL=\bl] = F_i(w_j|x_i,\mathbf{l}_i) + (1-F_i(w_j|x_i,\mathbf{l}_i)) \E[Y_i(t)|T_i \ge w_j, T_j=w_j,\X=\x,\mathbf{L}=\mathbf{l}] 
  \label{eq:id}
\end{equation}
otherwise, if $t \le w_j$, $\E[Y_i(t;w_j,\x,\mathbf{L}=\mathbf{l})] = F_i(t|x_i,\mathbf{l}_i)$. 
\label{thm:id}
\end{thm}

The structure of \eqref{eq:id} shows that the average controlled potential infection outcome $Y_i(t;w_j,\x)$ given $\bL=\bl$ is equal to the probability that $i$ is infected before $w_j$, plus the average outcome of $i$ when $j$ is infected at $w_j$ and $i$ is not yet infected, times the probability of this event. Figure \ref{fig:dag_swig} shows a causal graphical model \citep{pearl2009causality} of the system outlined by Assumptions 1-5, and the corresponding single-world intervention graph \citep{richardson2013single,richardson2013single2} for identification of the exposure-controlled potential infection outcome.

\begin{figure}
\centering
\includegraphics[width=\textwidth]{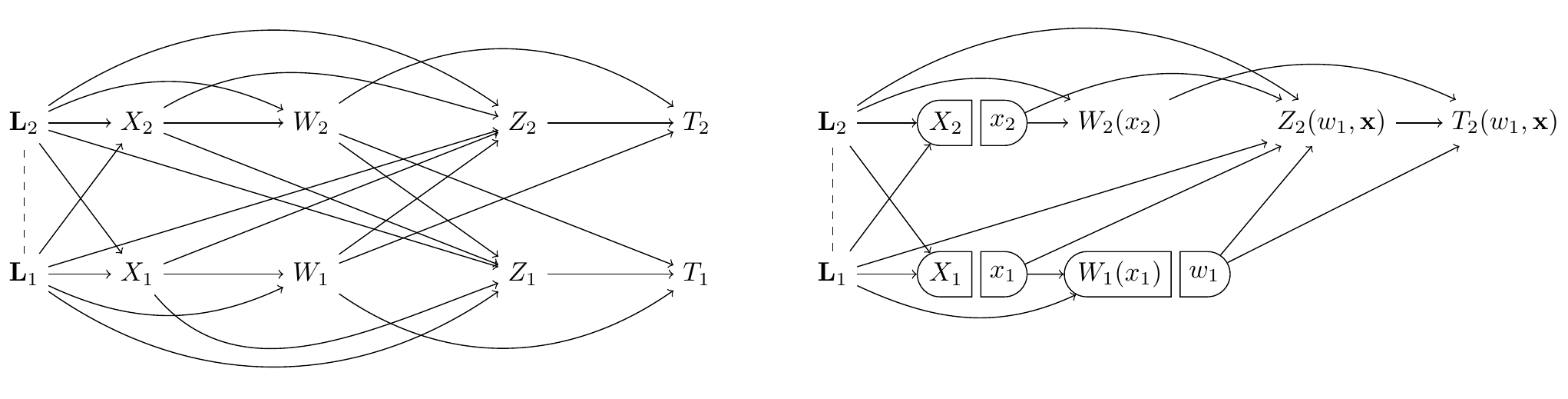}
\caption{Causal graphical model for infection outcomes in a two-person partnership, under Assumptions \ref{as:exclusion}--\ref{as:positivity}. At left, covariates $\bL_1$ and $\bL_2$ may be dependent within partnerships, and covariates of both subjects may affect the joint treatments $(X_1,X_2)$. The initial infection times $W_1$ and $W_2$ are functions of individual covariates and treatments alone by Assumption \ref{as:exclusion}, and thus no arrows exist from $X_j$ to $W_i$ or from $L_j$ to $W_i$. Subsequent infection times $Z_1$ and $Z_2$ are functions of treatments and covariates of both subjects, and the infection time of the first infected subject. From the decomposition of the infection time \eqref{eq:Tdef}, the infection time $Z_i$ is not a function of the (possibly latent) time $W_i$.  The overall infection time $T_i$ is determined by $W_i$, $W_j$ and $Z_i$, as specified in \eqref{eq:Tdef}. At right, the single-world intervention graph illustrates identification of $Y_1(t;w_2,\x)$ under the intervention in which $W_j=w_j$ and $\X=(x_1,x_2)$ are fixed. The variables $Z_2$ and $T_2$ are omitted since they are not relevant to generation of the outcome $T_1(w_1,\x)$.}
\label{fig:dag_swig}
\end{figure}

If we do not fix the infection time $W_j=w_j$, and instead allow it to take its ``natural'' value under a particular treatment to $j$, we obtain the marginal average potential infection outcome when $\bL=\bl$ as follows.

\begin{cor}[Identification of average natural potential infection outcome]
  Suppose Assumptions \ref{as:exclusion}-\ref{as:positivity} hold.  Then for $\x=(x_i,x_j)$, 
$\E[Y_i(t;W_j(x_j),\x)|\bL=\bl] = \E[Y_i(t)|\X=\x,\bL=\bl]$.  If in addition $x_j' \neq x_j$ and Assumption \ref{as:crossworld} holds,
\[ \E[Y_i(t;W_j(x_j'),\x)|\bL=\bl] = \int_0^t \E[Y_i(t;w_j,\x)|\bL=\bl] dF_j(w_j|x_j',\bl_j) . \]
where $F_j(w_j|x_j,\bl_j)$ is given by Lemma \ref{lem:Fi} and $\E[Y_i(t;w_j,\x)|\bL=\bl]$ by Theorem \ref{thm:id}.
\label{cor:idnatural}
\end{cor}

Finally, by standardization of the potential infection outcome across the distribution of covariates $\bL$, we can identify the average potential infection outcome.  Let $G(\bl)$ be the distribution function of the joint covariate vector $\bL=\bl$ in the population of partnerships. Then
\[ \E[Y_i(t;w_j,\x)] = \int \E[Y_i(t;w_j,\x)|\bL=\bl] dG(\bl) \]
and 
\[ \E[Y_i(t;W_j(x_j'),\x)] = \int \E[Y_i(t;W_j(x_j'),\x)|\bL=\bl] dG(\bl) \] 
where $\E[Y_i(t;w_j,\x)|\bL=\bl]$ and $\E[Y_i(t;W_j(x_j'),\x)|\bL=\bl]$ are given by Theorem \ref{thm:id} and Corollary \ref{cor:idnatural} respectively. 
\label{thm:adjustment}


\section{Causal estimands}

Contrasts of potential infection outcomes under different treatments $\x$ and infection times $w_j$ can yield epidemiologically meaningful causal estimands.  In this paper, we express causal contrasts as differences of average potential infection outcomes.  Effect measures on the hazard ratio, risk ratio, or odds ratio scales may be defined similarly \citep[e.g.][]{halloran1991direct,ohagan2014estimating}.  

First, the contagion effect captures the change in infection risk in one individual due to a change in the infection time of their partner.
\begin{defn}[Contagion effect] 
For $w_j \neq w_j'$ and treatment $\x=(x_i,x_j)$, the controlled contagion effect is 
$\CE(t,w_j,w'_j,\x) = \E[ Y_i(t;w_j,\x) - Y_i(t;w_j',\x)]$  
and the natural contagion effect is 
$\CE(t,\x) = \E[Y_i(t;W_j(0),\x) - Y_i(t;W_j(1),\x)]$.
\label{defn:contagion}
\end{defn}
We say that the infection outcome (absent treatment) is ``positively contagious'' if for  all infection times $w_j < w_j'$ with $w_j < t$, the controlled contagion effect under no treatment is $\CE(t,w_j,w_j',\mathbf{0}) > 0$.  In this way, we interpret contagion, or outcome transmissibility, as a causal phenomenon that need not depend on treatments: under positive contagion, earlier infection of one's partner \emph{causes} one to become infected earlier, on average.  On the other hand, the natural contagion effect $\CE(t,\x)$ captures features of the treatment effect: it replaces fixed values of $w_j$ and $w_j'$ with their counterfactual values under no treatment versus treatment of $j$, similar to the effect proposed by \citet{vanderweele2012components} for an asymmetric partnership.  The natural contagion effect is a ``cross-world'' estimand because it integrates the average potential infection outcome $\E[Y_i(t;w_j,\x)]$ with respect to the distribution of $W_j(x_j')$ under a treatment $X_j=x_j'$ that cannot arise in the same realization as $X_j=x_j$. Figure \ref{fig:illustration} can be reinterpreted in light of Definition \ref{defn:contagion}: positive contagion means that earlier infection of $j$ causes $i$ to become infected earlier, compared to the infection time of $i$ that would have occurred, had $W_j$ happened later. 

The susceptibility effect is of interest in vaccine trials because it summarizes the clinical effect of an intervention on the individual who receives it, holding the treatment status and infection time of their partner constant \citep{halloran1995causal,halloran1997study,golm1999semiparametric}.  The susceptibility effect is sometimes called the ``per-exposure effect'' because it holds the distribution of exposure to infectiousness constant \citep{ohagan2014estimating}. 
\begin{defn}[Susceptibility effect]
For $w_j>0$ and $X_j=x_j$, the controlled susceptibility of effect is 
$\SE(t,w_j,x_j) = \E[ Y_i(t;w_j,x_i=1,x_j) - Y_i(t;w_j,x_i=0,x_j)]$
and the natural susceptibility effect is 
$\SE(t,x_j) = \E[ Y_i(t;W_j(x_j),x_i=1,x_j) - Y_i(t;W_j(x_j),x_i=0,x_j) ]$.
\label{defn:susceptibility}
\end{defn}
If the controlled susceptibility effect is negative for every $w_j$ and $x_j$, this means that the treatment is beneficial to the individual who receives it.  Note that the natural susceptibility effect is not a cross-world estimand: it averages potential infection outcomes with respect to the distribution of $W_j(x_j)$, where $x_j$ is the treatment under which the infection outcome of $i$ is realized. 

The infectiousness effect summarizes the effect of changing the treatment to $j$ on the infection risk of $i$, while holding the treatment to $i$ and the infection time of $j$ unchanged.  
\begin{defn}[Infectiousness effect] 
For $w_j>0$ and $X_i=x_i$, the controlled infectiousness effect is $\IE(t,w_j,x_i) = \E[ Y_i(t;w_j,x_i,x_j=1) - Y_i(t;w_j,x_i,x_j=0)]$  
and the natural infectiousness effect is $\IE(t,x_i) = \E[Y_i(t;W_j(0),x_i,x_j=1) - Y_i(t;W_j(0),x_i,x_j=0)]$.  
\label{defn:infectiousness}
\end{defn}
The natural infectiousness effect is a cross-world estimand because the first term in the contrast specifies that the infection time of $j$ is realized under $x_j=0$, but the infectiousness of $j$ subsequently is realized under $x_j=1$.  Several authors have described the natural infectiousness effect as unidentified even under randomization when only binary infection outcomes are recorded at follow-up \citep{vanderweele2011bounding,chiba2012note,halloran2012causal,vanderweele2012components,chiba2012note,chiba2013simple,chiba2013conditional}.  Corollary \ref{cor:idnatural} and Definition \ref{defn:infectiousness} together show why this is the case.  The correct marginalization over infection times $W_j(x_j')$ cannot be computed unless the distribution of $W_j(x_j)$ is identified as in Lemma \ref{lem:Fi}.  The controlled and natural infectiousness effects are similar to those proposed by \citet{chiba2013conditional}, but here the marginalization is over the infection time of $j$, not their binary infection outcome. 

\section{Comparison to other infectious disease intervention effects}

Statisticians and epidemiologists have proposed a wide variety of estimands summarizing the effect of interventions for contagious outcomes, often in the two-person partnership setting.  In this section, we evaluate the meaning of alternative definitions of vaccine effect estimands in the context of the causal effects defined above.  We take the contagion, susceptibility, and infectiousness effects defined above as fundamental characteristics of the disease transmission process and intervention under study.  Whenever possible, we characterize the sign, or direction, of alternative effects, as a function  of these primitives.  In some cases, where the relationship is complex, we evaluate the alternative estimands under a null hypothesis, for example when the controlled susceptibility or infectiousness effect is zero, so that explicit results can be analytically proven.  For simplicity, we omit the role of covariates $\bL$ in the comparison of estimands. 

The ``attack rate'' of an infectious disease is defined for individuals with treatment $x$ as $\AR_x(t) = \E[Y_i(t)|X_i=x]$.  The ratio of attack rates, sometimes called ``relative cumulative incidence'', has been used to measure the vaccine effect on susceptibility, defined as $\VE_{\AR}(t) = 1- \AR_1(t)/\AR_0(t)$ \citep{greenwood1915statistics,francis1955evaluation}.  A related estimand, called the ``direct effect'', is a contrast on the difference scale, $\DE(t) = \AR_1(t) - \AR_0(t)$ when treatment is randomized within groups \citep{hudgens2008toward}.  In the symmetric partnership setting, attack rates $\AR_x(t)$ that condition only on the treatment to $i$ implicitly marginalize over treatment to $j$. 
\begin{thm}
Suppose $\SE(t,w_j,x_j) = 0$ and $\IE(t,w_j,x_j)<0$ for all $x_j$ and $w_j>0$.  If $\X=(X_i,X_j)$ is positively dependent, then $\DE(t)<0$ and $\VE_{\AR}(t)>0$; if $\X$ is negatively dependent then $\DE(t)>0$ and $\VE_{\AR}(t)<0$; and if $X_i \indep X_j$ then $\DE(t)= \VE_{\AR}(t)=0$.  If there is no treatment effect whatsoever, $\SE(t,w_j,x)=\IE(t,w_j,x)=0$ for all $x$ and $w_j>0$, then $\DE(t)=\VE_{\AR}(t)=0$ for any joint distribution of $\X$. 
  \label{thm:VEAR}
\end{thm}
In other words, $\VE_{\AR}(t)$ and $\DE(t)$ may or may not recover the sign, or direction, of the susceptibility effect, depending on the susceptibility and infectiousness effects, and the joint distribution of $\X$ within clusters.  
\citet{morozova2018risk} and \citet{eck2019randomization} proved similar results in a parametric setting under Bernoulli, block, and cluster randomization for the joint treatment $\X$ in clusters or partnerships.
\citet{halloran1991direct}, \citet{halloran1994exposure}, \citet{halloran1997study} and \citet{rhodes1996counting}  warned that $\VE_{\AR}(t)$ may be a biased approximation to the susceptibility effect due to differential exposure to infection between treated and untreated individuals in clusters.  

Related definitions of the attack rate condition on the treatments to both individuals in the partnership.  The attack rate among individuals with treatment $x$ whose partner has treatment $x'$ is $\AR_{x,x'}(t) = \E[Y_i(t)|X_i=x,X_j=x']$.  The indirect effect is defined as $\IDE(t) = \AR_{01}(t) - \AR_{00}(t)$ \citep{hudgens2008toward}, and is equivalent to the difference of the natural infectiousness and contagion effects defined above: 
\begin{equation*}
\begin{split}
\IDE(t) & = \E[Y_i(t;W_j(1),(0,1)) - Y_i(t;W_j(0),\mathbf{0})] \\
       & = \E[Y_i(t;W_j(1),(0,1)) - Y_i(t;W_j(0),(0,1))] + \E[Y_i(t;W_j(0),(0,1)) - Y_i(t;W_j(0),\mathbf{0})] \\
       &= -\CE(t,(0,1))+\IE(t,0) .
\end{split}
\end{equation*}
The secondary attack rate is defined as $\SAR_{x',x}(t) = \E[Y_i(t)|T_j<t,T_i>T_j,X_i=x,X_j=x']$, the average infection status of $i$ when $j$ is infected during the study before $i$, under treatments $x$ and $x'$ to $i$ and $j$ respectively.  The ``secondary attack rate for infectiousness'' is defined as $\VE_{\I}^{\net}(t) = 1 - \SAR_{10}(t)/\SAR_{00}(t)$ \citep{halloran2012causal}. 
We analyze $\VE_{\I}^{\net}(t)$ under the null hypothesis of no infectiousness effect, and show that when the infection is contagious and there is a susceptibility effect, $\VE_{\I}^{\net}(t)$ may nevertheless be nonzero.  Let $h_0(u|0)$ be the hazard of of the potential infection time $W_i(0)$, and let $h_0(u|1)$ be the hazard of $W_i(1)$. 
\begin{thm}
Suppose $\IE(t,w_j,0) = 0$, $\CE(t,w_j,w'_j,\mathbf{0})>0$ for all $0<w_j<w'_j$, and $h_0(u|1)=\varepsilon h_0(u|0)$ with $\varepsilon \in [0,1)$, then $\VE_{\I}^{\net}(t) > 0$. Suppose $\IE(t,w_j,0) = \SE(t,w_j,x_j) = 0$ for all $w_j$ and $x_j$, then $\VE_{\I}^{\net}(t) = 0$. Suppose $\CE(t,w_j,w'_j,\mathbf{0})=0$ for all $0<w_j<w'_j$ and $h_0(u|1)=\varepsilon h_0(u|0)$ with $\varepsilon \in [0,1)$, then $\VE_{\I}^{\net}(t) > 0$.
\label{thm:VEInet}
\end{thm}
In other words, when the true infectiousness effect is null, the infection outcome is positively contagious, and the vaccine has a favorable susceptibility effect prior to the first infection, $\VE_{\I}^{\net}(t)$ is nonzero. In a more extreme case, when the true contagion effect is null, the disease is not transmissible; if the vaccine has a favorable susceptibility effect prior to the first infection, then $\VE_{\I}^{\net}(t)$ is still nonzero. 

A simple explanation shows why $\VE_{\I}^{\net}$ can behave in unexpected ways: is not solely a function of the infectiousness effect. Instead, $\VE_{\I}^{\net}(t)$ also incorporates features of the susceptibility effect on the partner of $i$ \emph{before the first infection occurs}. When both the infectiousness and susceptibility effects are null, $\VE_{\I}^{\net}(t)$ recovers the correct null effect of infectiousness. Several authors have pointed out that $\VE_{\I}^{\net}(t)$ may suffer from selection bias because it conditions on post-randomization variables -- the infection status of both partners \citep{halloran1991direct,halloran1994exposure,rhodes1996counting,halloran1997study,halloran2012causal}.  \citet{halloran2012causal} use tools from principal stratification to derive bounds for the infectiousness effect under this selection bias, and propose a bound estimator ${\CVE}_{\I}^c(t)$ for $\VE_{\I}^{\net}(t)$ under Bernoulli randomization. We analyze these bounds by simulation below.

Several authors have recognized that simple comparison of outcomes in treated versus untreated individuals may not suffice  to identify meaningful causal effects for infectious disease interventions, even under randomization.  For example, \citet{vanderweele2010direct}, \citet{vanderweele2012components}, \citet{ogburn2014causal}, and \citet{ogburn2017vaccines} apply tools from mediation analysis to a simplified partnership model to identify contagion and infectiousness effects similar to those we have defined above.  In the asymmetric partnership model, subject $i$ is unvaccinated, home-bound, and can only be infected by their (possibly vaccinated) partner $j$.  To represent this structural assumption in the framework outlined here, we force the infection time of the home-bound partner, in the absence of infection in their partner, to be infinite.  To this end, let hazard of $W_i(0)$ be $h^i_0(t|0)=0$, so that infection of $i$ from an external source can never occur.  These authors define the infectiousness effect as $\VE_{\I}(t) = \E[Y_i(t;Y_j(1),(0,1))] -\E[Y_i(t;Y_j(1),(0,0))]$, which contrasts the infection outcomes of $i$ when $j$ is treated versus untreated, with $j$'s infection status $Y_j(x_j)$ set to the value it would take if $j$ were treated. 
\begin{thm}
Suppose $h^i_0(t|0)=0$ for all $t>0$. Then $\VE_{\I}(t)=\IE(t,0)$.
\label{th:asy_infectiousness}
\end{thm}
In other words, under the asymmetric setting where $i$ is unvaccinated and cannot be infected from outside the partnership, $\VE_{\I}(t)$ is equivalent to the natural infectiousness effect in Definition \ref{defn:infectiousness}. 

A contagion effect is defined by \citet{vanderweele2012components} as $\VE_{\C}(t) = \E[Y_i(t;Y_j(1),(0,0)) - Y_i(t;Y_j(0),(0,0))]$. Note that this quantity reverses the difference in Definition \ref{defn:contagion}, contrasting the infection outcome of $i$ when the infection status of $j$ is set to the value it would obtain if $j$ were treated versus untreated. We provide sufficient conditions for $\CE(t,u,u',\mathbf{0})$ and $\VE_{\C}(t)$ to behave similarly, that is, to have opposite sign.  
\begin{thm}
  Suppose $h^i_0(t|0)=0$ for all $t>0$ and $\SE(t,w_j,0)>0$. Then $\VE_{\C}(t)$ has opposite sign as $\CE(t, w_j,w'_j,\mathbf{0})$ for $0<w_j<w'_j$. Suppose $h^i_0(t|0)=0$, $\SE(t,w_j,0)=0$ and $\CE(t,w_j,w'_j,\mathbf{0})>0$, then $\VE_{\C}(t)=0$.
\label{th:asy_contagion}
\end{thm}
In other words, in the asymmetric partnership setting, $-\VE_{\C}$ recovers the sign of the contagion effect in Definition \ref{defn:contagion} when the vaccine has a favorable susceptibility effect \emph{prior to the first infection}.  However, if the true susceptibility effect is null, $\VE_{\C}(t)=0$ regardless of the true contagion effect.  That is, $\VE_{\C}(t)$ is largely a measure of the susceptibility effect, and not of the transmissibility of the infection outcome. 


\section{Application: a hypothetical vaccine trial}

We simulate observational and randomized trials of a hypothetical HIV vaccine in a large population of sexual partnerships \citep{halloran1994exposure}.  We assume individuals are not infected at baseline, but that either individual may become infected from outside the partnership, and transmission within partnerships may occur.  To parameterize the infection transmission process, we specify hazard models for the infection times $W_i(x_i)$ and $Z_i(w_j,\x)$. This approach has been employed in extensive prior work on statistical models for time-to-infection data \citep{halloran1994exposure,rhodes1996counting,kenah2008generation,kenah2010contact,kenah2013non,kenah2015semi}.  For a time $t>0$, Let the hazard of $W_i(x_i)$ given covariates $\bL_i=l_i$ be given by 
\begin{equation}
\lambda_i^W(t;x_i,l_i) = \alpha(t) e^{\beta_0 x_i + \theta_0'\bl_i} .
\label{eq:lamiW}
\end{equation}
In words, the hazard of infection in an individual whose partner is not infected, is given by a Cox model with baseline hazard $\alpha(t)$.  Following infection of $j$ at time $W_j=w_j$, the remaining potential infection time $Z_i(w_j,\x)$ given $\bL=\bl=(l_i,l_j)$ has hazard
\begin{equation}
\lambda_i^Z(t;w_j,\x,\bl) = \lambda_i^W(t;x_i) + \gamma(t-w_j) e^{\beta_1 x_i + \sigma x_j + \theta_1'\bl_j + \theta_2'\bl_i} 
\label{eq:lamiZ}
\end{equation}
for $t>w_j$. The coefficients $\beta_0$ and $\beta_1$ represent the change in infection risk due to vaccination of $i$, and $\sigma$ represents the change in transmission risk due to vaccination in $j$ when $j$ is infected. Covariate effects are represented by $\theta_0$, $\theta_1$, and $\theta_2$, and $\alpha(t)$ and $\gamma(t-w_j)$ are baseline transmission hazards for the external and internal forces of infection respectively.  This specification implies that the external force of infection and transmissibility are \emph{competing risks} for infection of $i$ \citep{longini1982estimating,longini1988statistical,rampey1992discrete}.  That is, a susceptible individual can be infected by a source of infectiousness outside their partnership, or from an infected partner.  When the baseline hazards $\alpha(t)$ and $\gamma(t-w_j)$ are time-invariant, the model reduces to a Markov susceptible-infective process with an external force of infection \citep[e.g][]{morozova2018risk,eck2019randomization}. We consider both time-invariant baseline infection hazards $\alpha(t)=\alpha$ and $\gamma(t)=\gamma$, and the seasonally varying $\alpha(t) = a(1+\sin(2 \pi t+\phi))$ with decaying infectiousness from an infected partner $\gamma(t-w_j) = b\exp[-\omega(t-w_j)]$.  For any functional forms of the time-varying baseline hazards $\alpha(t)$ and $\gamma(t-w_j)$, the hazard models \eqref{eq:lamiW} and \eqref{eq:lamiZ} imply distributions for $W_i(x_i)$ and $Z_i(w_j,\x)$, and hence $T_i(w_j,\x)$, that obey the identification Assumptions \ref{as:txignorability}--\ref{as:crossworld}.

Subjects in partnerships are endowed with individual characteristics $\bL$ that may be correlated. In the randomized trial simulation, the vaccine is randomized in accordance with a specified distribution -- Bernoulli, block, or cluster randomization -- without regard to these traits.  Under each randomization design, the marginal treatment probability $\Pr(X_i=x_i)$ is 1/2.  For Bernoulli randomization, $\Pr(\X=\x)=1/4$, for block randomization, $\Pr(\X=\x) = \indicator{\sum_i x_i = 1}/2$, and for cluster randomization, $\Pr(\X=(1,1)) = 1/2$ and $\Pr(\X=(0,0)) = 1/2$.  In the observational study simulation, the traits $\bL=\bl$ together determine the joint distribution of vaccine in the partnership as 
\[ \Pr(X_i=1|L_i=l_i) = \frac{1}{1+e^{-l_i}} \]
where 
\[ \begin{pmatrix}L_i \\ L_j\end{pmatrix} \sim \text{Normal}\left(\begin{pmatrix} 0 \\ 0 \end{pmatrix}, v\begin{pmatrix} 1 & \rho\\ \rho & 1 \end{pmatrix}\right)  \]
with $v>0$. 

Figure \ref{allcontrasts} illustrates controlled infection outcomes $\E[Y_i(t;w_j,\x)]$ over time for different choices of $w_j$ and $\x$, under the time-invariant hazard models. Contrasts of these potential infection outcomes give controlled contagion, susceptibility and infectiousness effects, shown in the lower-right corner of Figure \ref{allcontrasts}. 

Tables \ref{tab:sim} and \ref{tab:sim2} show estimates of the natural contagion, susceptibility and infectiousness effects from simulated data, and compare these values to alternative estimands proposed by other authors, including  the direct effect $\DE(t)$, the indirect effect $\IDE(t)$, the secondary attack rate infectiousness effect $\VE_{\I}^{\net}(t)$, and ${\CVE}_{\I}^c(t)$ bounds introduced by \citep{halloran2012causal}. All natural or marginal estimands are evaluated at time $t=2$ years under each design and under both time-invariant and time-varying baseline hazards with a sample size $N=100{,}000$ partnerships. Estimands that are not identified under a given design are not evaluated. Table \ref{tab:sim} shows a simulation example where the estimated direct effect $\DE(t)$ is positive (0.06 and 0.08) under block randomization when the disease is contagious, even though the true susceptibility effect is negative, or beneficial \citep[see, e.g.][]{morozova2018risk,eck2019randomization}. Table \ref{tab:sim2} shows another simulation setting where $\DE(t)$ achieves the same sign as the susceptibility effect. But in the setting without contagion, the disease is not contagious and infection outcomes are realized independently. Therefore, all ``indirect'' and ``infectiousness'' effects should be null. However, $\VE_{\I}^{\net}(t)$ is negative (-0.01 and -0.02 in both Table \ref{tab:sim} and \ref{tab:sim2}), conflicting with the fact that the disease is not transmissible (as proved in Theorem \ref{thm:VEInet}). The identification interval $\CVE_{\I}^c(t)$ has nonzero width, but covers zero. Figure \ref{differentestimand} compares estimates of different types of natural susceptibility and infectiousness effects over time, when both effects are beneficial (negative). In the bottom-right panel of Figure \ref{differentestimand}, we show that $\DE(t)$ under block randomization can suffer from directional bias.

\begin{figure}
    \centering 
\begin{subfigure}{0.49\textwidth}
  \includegraphics[width=\linewidth]{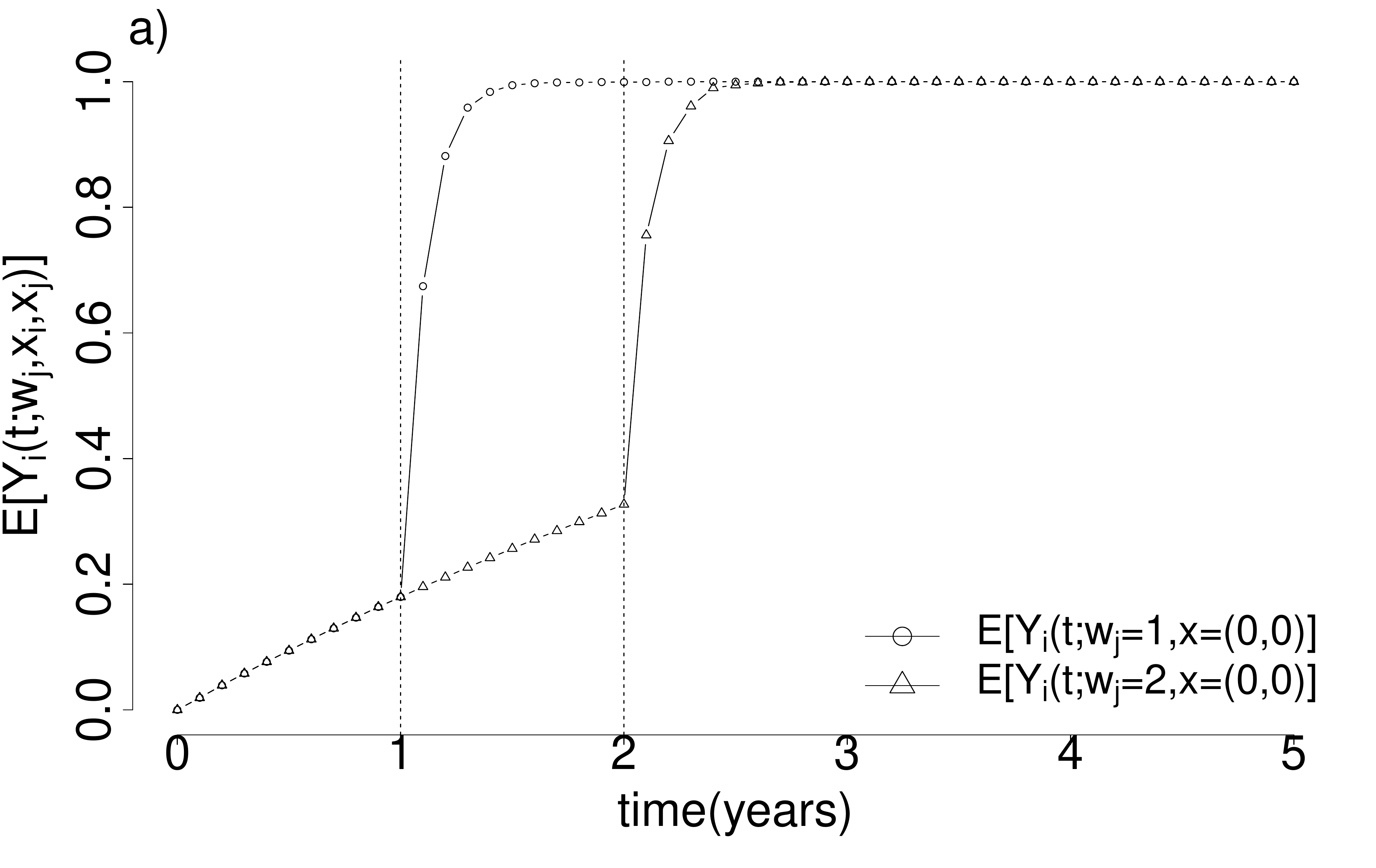}
\end{subfigure}\hfil 
\begin{subfigure}{0.49\textwidth}
  \includegraphics[width=\linewidth]{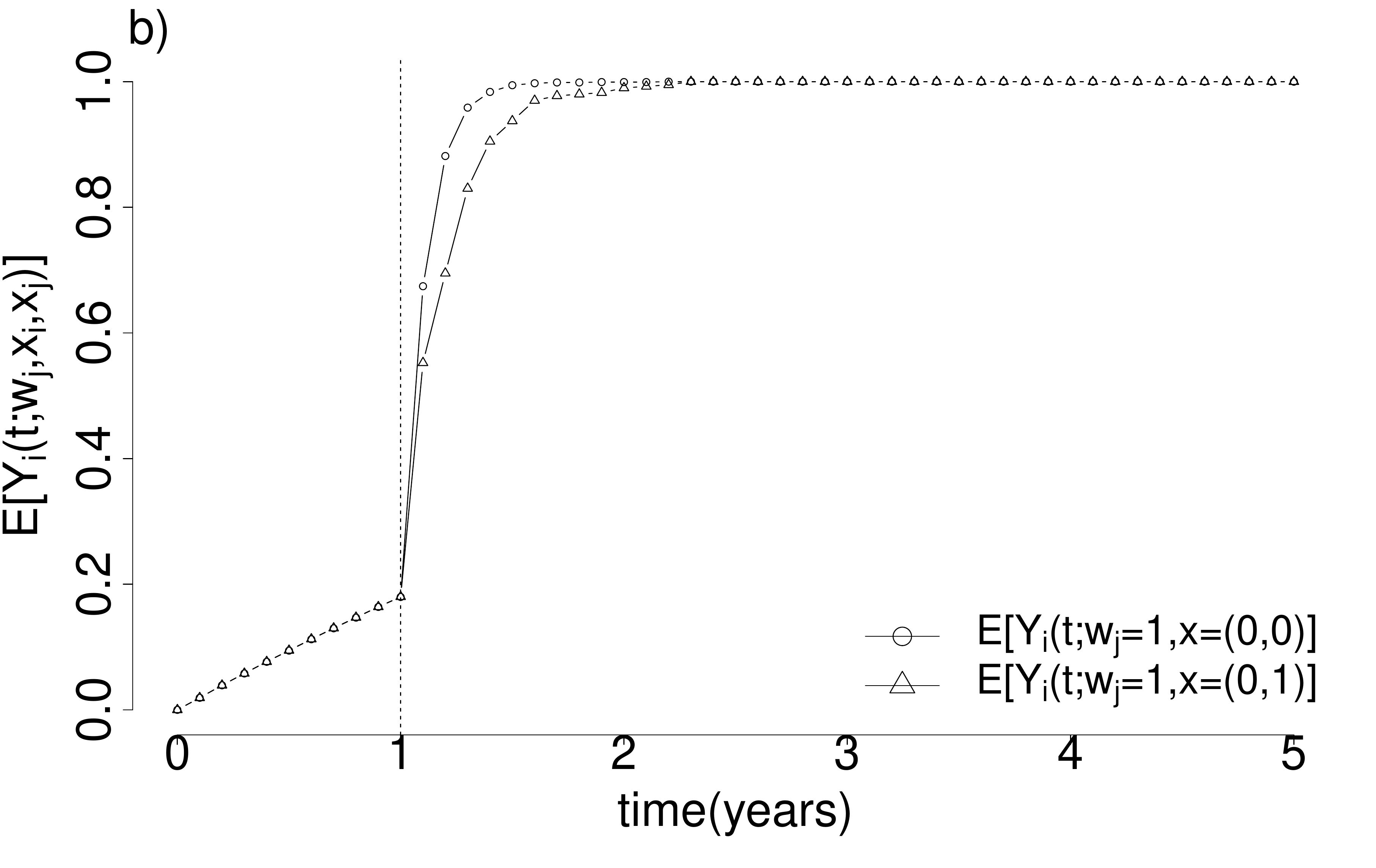}
\end{subfigure}\hfil 
\medskip
\begin{subfigure}{0.49\textwidth}
  \includegraphics[width=\linewidth]{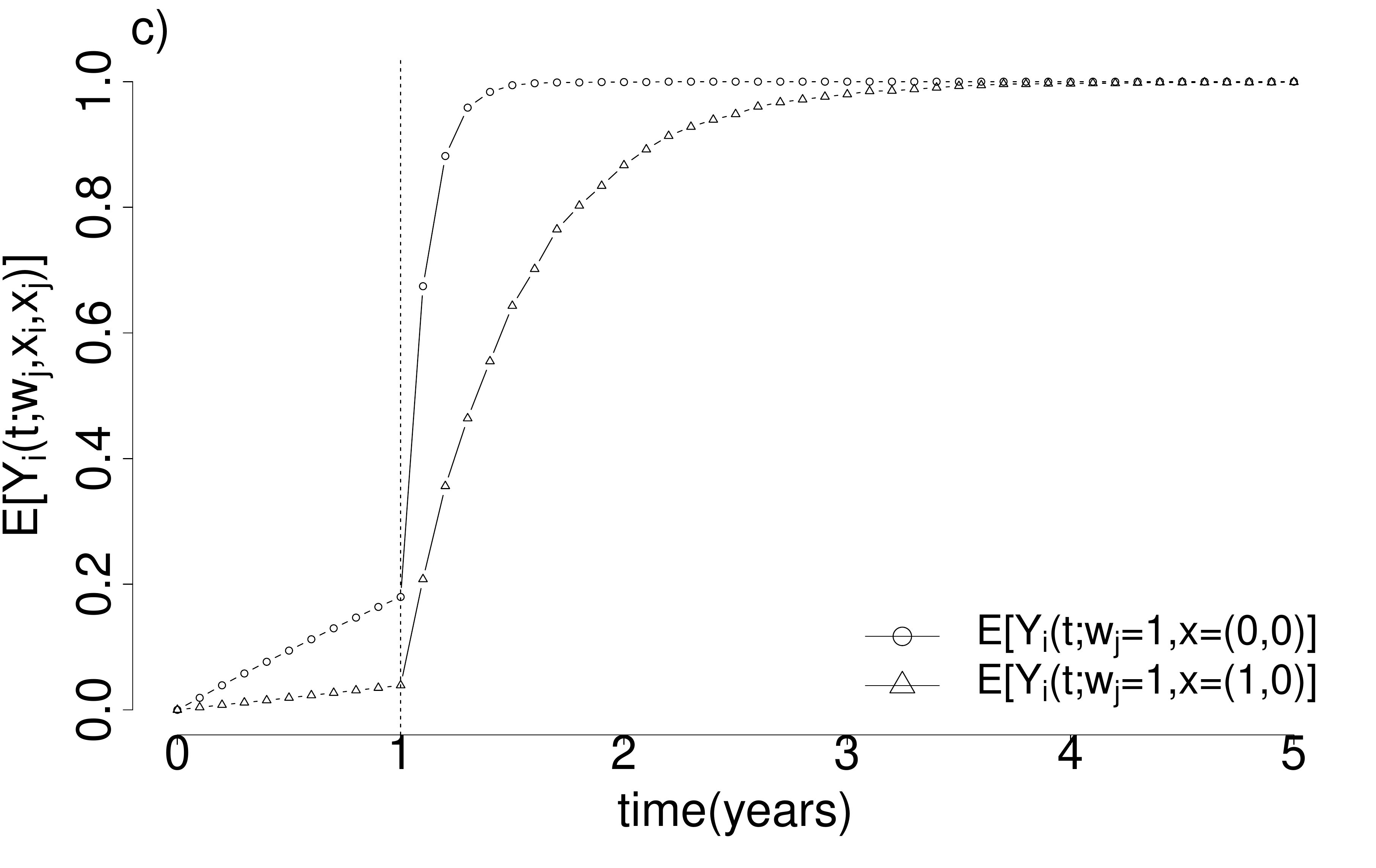}
\end{subfigure}\hfil 
\begin{subfigure}{0.49\textwidth}
  \includegraphics[width=\linewidth]{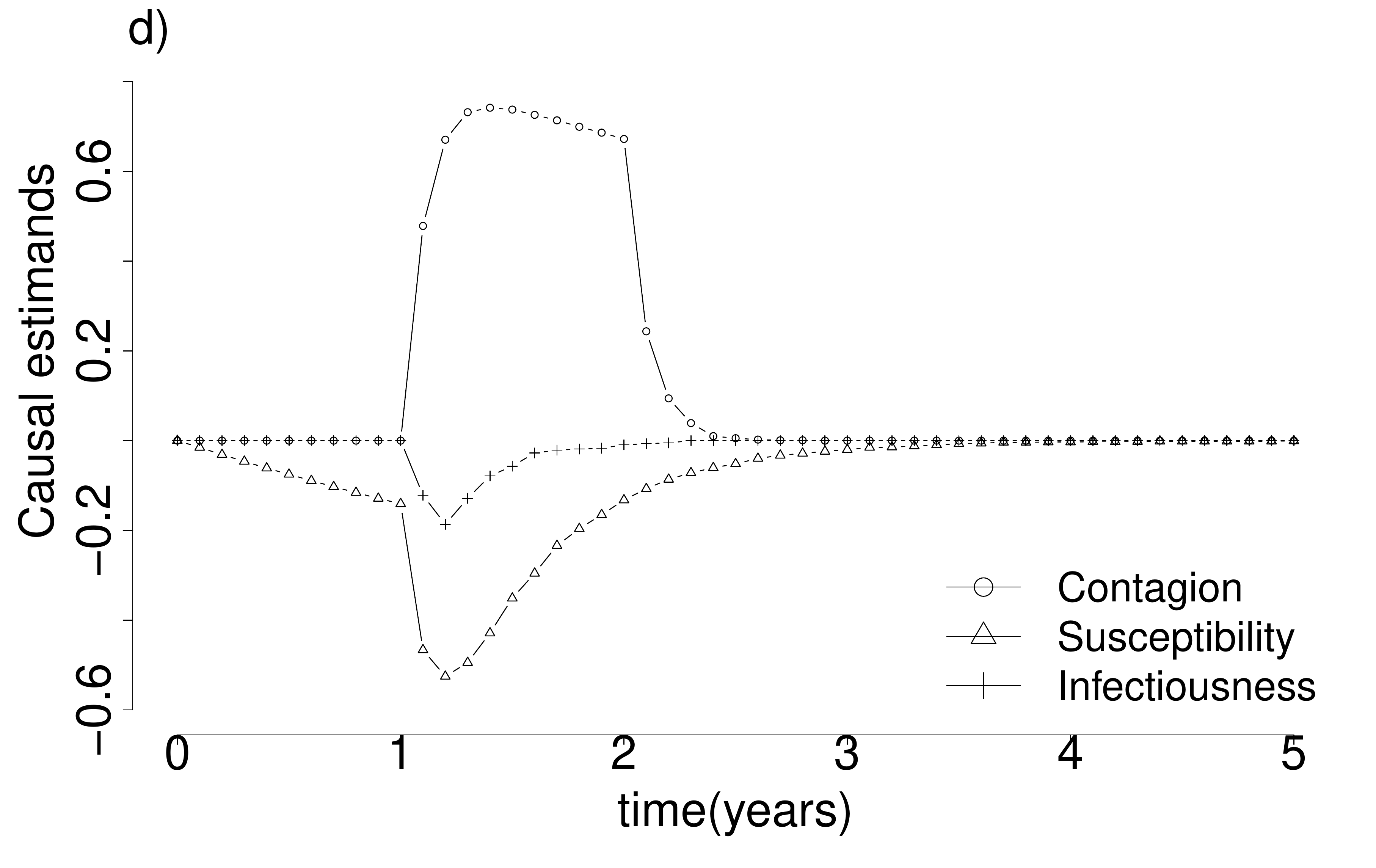}
\end{subfigure}\hfil 
\caption{Illustration of average controlled potential infection outcomes under different values of the infection time $w_j$ and joint treatment $\x$, under time-invariant baseline hazards $\alpha(t)=0.2$ and $\gamma(t-w_j)=10$ and coefficients $e^{\beta_0}=e^{\beta_1}=0.2$ and $e^{\sigma}=0.5$. Contrasts of potential outcomes in (a), (b) and (c) show the controlled contagion effect, the infectiousness effect, and the susceptibility effect evaluated at different times, shown together in (d).} 
\label{allcontrasts}
\end{figure}

\begin{sidewaystable}
    \centering
\begin{tabular}{llccccccc}
\hline
Simulation & Treatment & $\CE(t,\mathbf{0})$ & $\SE(t,0)$ & $\IE(t,0)$ & $\DE(t)$ & $\IDE(t)$ & $\VE_{\I}^{\net}(t)$ & ${\CVE}_{\I}^{c}(t)$ \\
  \hline
Constant hazards & Obs. & 0.12 & -0.14 & -0.19 & -0.16 & -0.20 & -0.68 &  \\ 
   & Bernoulli & 0.12 & -0.14 & -0.19 & -0.16 & -0.20 & -0.69 & ( -0.73  , -0.66 ) \\ 
   & Block &  &  &  & 0.06 &  &  &  \\ 
   & Cluster &  &  &  & -0.39 &  &  &  \\[5pt] 
  Constant hazards & Obs. & 0.00 & -0.17 & 0.00 & -0.17 & 0.00 & -0.01 &  \\ 
  without contagion & Bernoulli & 0.00 & -0.18 & 0.00 & -0.18 & 0.00 & -0.01 & ( -0.25  , 0.19 ) \\ 
   & Block &  &  &  & -0.18 &  &  &  \\ 
   & Cluster &  &  &  & -0.18 &  &  &  \\ [5pt]
  Time-varying hazards & Obs. & 0.11 & -0.14 & -0.19 & -0.22 & -0.21 & -0.50 &  \\ 
   & Bernoulli & 0.12 & -0.14 & -0.20 & -0.21 & -0.22 & -0.52 & ( -0.53  , -0.5 ) \\ 
   & Block &  &  &  & 0.08 &  &  &  \\ 
   & Cluster &  &  &  & -0.50 &  &  &  \\[5pt] 
  Time-varying hazards & Obs. & 0.00 & -0.28 & 0.00 & -0.28 & 0.00 & -0.02 &  \\ 
  without contagion & Bernoulli & 0.00 & -0.28 & 0.00 & -0.28 & 0.00 & -0.02 & ( -0.42  , 0.36 ) \\ 
   & Block &  &  &  & -0.28 &  &  &  \\ 
   & Cluster &  &  &  & -0.28 &  &  &  \\ 
   \hline
\end{tabular}
\caption{Simulation results showing estimates of the natural contagion, susceptibility, infectiousness effects, and alternative estimands defined by \citet{hudgens2008toward}, \citet{halloran2012causal}, and \citet{vanderweele2012components}. Estimands are evaluated under four different scenarios -- the constant hazards ($\alpha(t)=0.2$ and $\gamma(t)=10$), constant hazards without contagion ($\alpha(t)=0.2$ and $\gamma(t)=0$), time-varying hazards ($\alpha(t)=0.4(1+\sin(2 \pi t+\frac{\pi}{2}))$ and $\gamma(t)=25e^{-0.5(t-t_j)}$), and time-varying hazards without contagion ($\alpha(t)=0.4(1+\sin(2\pi t+\frac{\pi}{2}))$ and $\gamma(t)=0$), respectively. The effect of vaccination is the same across all scenarios with $e^{\beta_0}=e^{\beta_1}=0.4$ and $e^{\sigma}=0.01$. The individual covariates $(l_i,l_j)$ are correlated with $\rho=0.1$ and coefficients of $e^{\theta_0}=e^{\theta_1}=e^{\theta_2}=0.95$.  }
\label{tab:sim}
\end{sidewaystable}

\begin{sidewaystable}
    \centering
\begin{tabular}{llccccccc}
\hline
Simulation & Treatment & $\CE(t,\mathbf{0})$ & $\SE(t,0)$ & $\IE(t,0)$ & $\DE(t)$ & $\IDE(t)$ & $\VE_{\I}^{\net}(t)$ & ${\CVE}_{\I}^{c}(t)$ \\
  \hline
Constant hazards & Obs. & 0.14 & -0.17 & -0.01 & -0.19 & -0.14 & -0.05 &  \\ 
   & Bernoulli & 0.14 & -0.18 & -0.01 & -0.20 & -0.14 & -0.04 & ( -0.08  , 0.02 ) \\ 
   & Block &  &  &  & -0.03 &  &  &  \\ 
   & Cluster &  &  &  & -0.36 &  &  &  \\ [5pt]
  Constant hazards & Obs. & 0.00 & -0.21 & 0.00 & -0.21 & 0.00 & -0.01 &  \\ 
  without contagion & Bernoulli & 0.00 & -0.22 & 0.00 & -0.22 & -0.00 & -0.01 & ( -0.38  , 0.19 ) \\ 
   & Block &  &  &  & -0.22 &  &  &  \\ 
   & Cluster &  &  &  & -0.22 &  &  &  \\ [5pt]
  Time-varying hazards & Obs. & 0.13 & -0.18 & -0.01 & -0.23 & -0.14 & -0.03 &  \\ 
   & Bernoulli & 0.15 & -0.18 & -0.01 & -0.23 & -0.15 & -0.03 & ( -0.04  , 0 ) \\ 
   & Block &  &  &  & -0.03 &  &  &  \\ 
   & Cluster &  &  &  & -0.44 &  &  &  \\ [5pt]
  Time-varying hazards & Obs. & -0.01 & -0.34 & 0.00 & -0.34 & 0.00 & -0.02 &  \\ 
  without contagion & Bernoulli & 0.00 & -0.34 & 0.00 & -0.34 & 0.00 & -0.02 & ( -0.65  , 0.35 ) \\ 
   & Block &  &  &  & -0.34 &  &  &  \\ 
   & Cluster &  &  &  & -0.34 &  &  &  \\ 
   \hline
\end{tabular}
\caption{Simulation results showing estimates of the natural contagion, susceptibility, infectiousness effects, and alternative estimands defined by \citet{hudgens2008toward}, \citet{halloran2012causal}, and \citet{vanderweele2012components}. Estimands are evaluated under four different scenarios -- the constant hazards ($\alpha(t)=0.2$ and $\gamma(t)=10$), constant hazards without contagion ($\alpha(t)=0.2$ and $\gamma(t)=0$), time-varying hazards ($\alpha(t)=0.4(1+\sin(2 \pi t+\frac{\pi}{2}))$ and $\gamma(t)=25e^{-0.5(t-t_j)}$), and time-varying hazards without contagion ($\alpha(t)=0.4(1+\sin(2\pi t+\frac{\pi}{2}))$ and $\gamma(t)=0$), respectively. The effect of vaccination is the same across all scenarios with $e^{\beta_0}=e^{\beta_1}=0.2$ and $e^{\sigma}=0.5$. The individual covariates $(l_i,l_j)$ are correlated with $\rho=0.1$ and coefficients of $e^{\theta_0}=e^{\theta_1}=e^{\theta_2}=0.95$.  }
\label{tab:sim2}
\end{sidewaystable}

\begin{figure}
    \centering 
\begin{subfigure}{0.49\textwidth}
  \includegraphics[width=\linewidth]{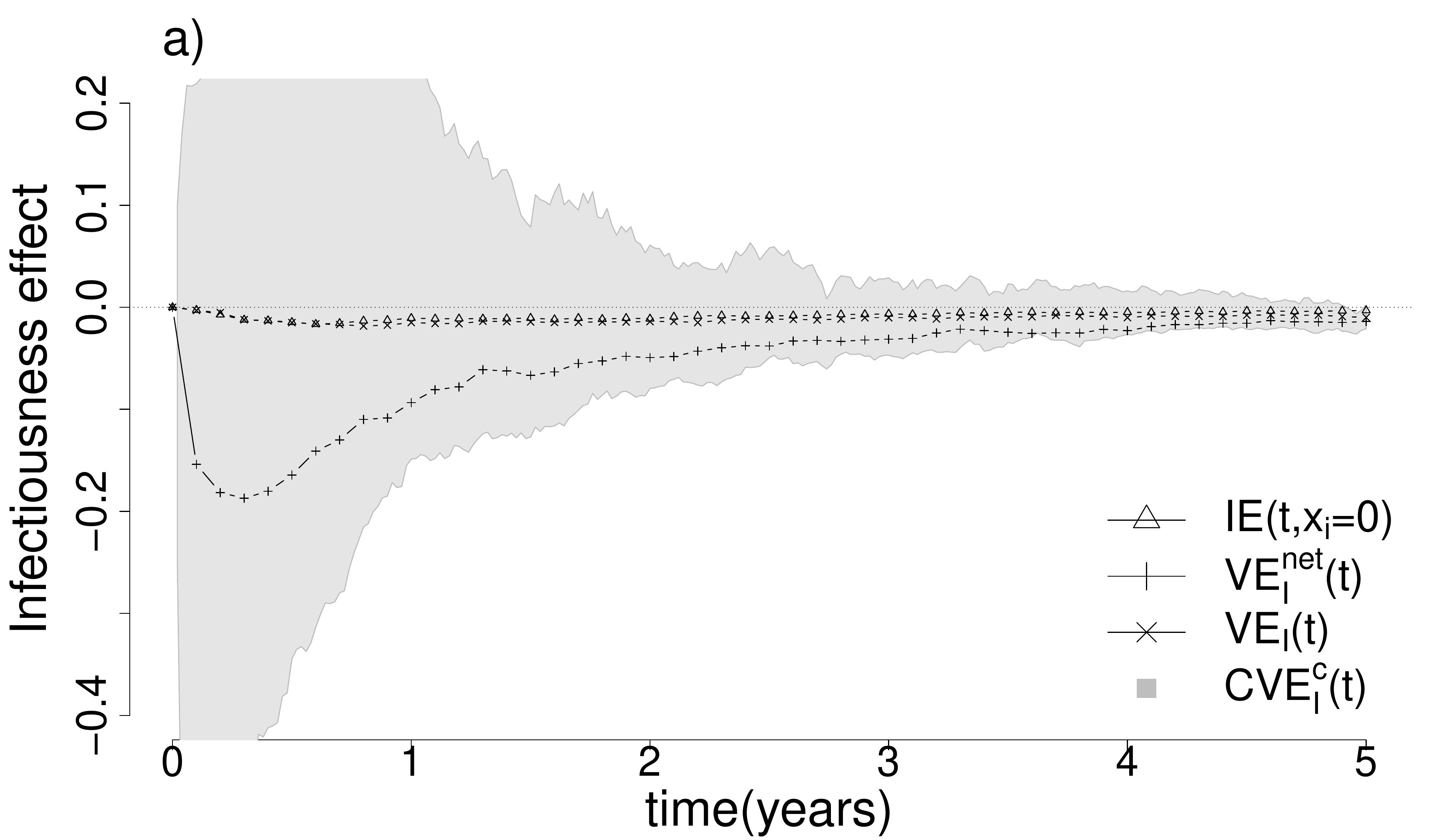}
\end{subfigure}\hfil 
\begin{subfigure}{0.49\textwidth}
  \includegraphics[width=\linewidth]{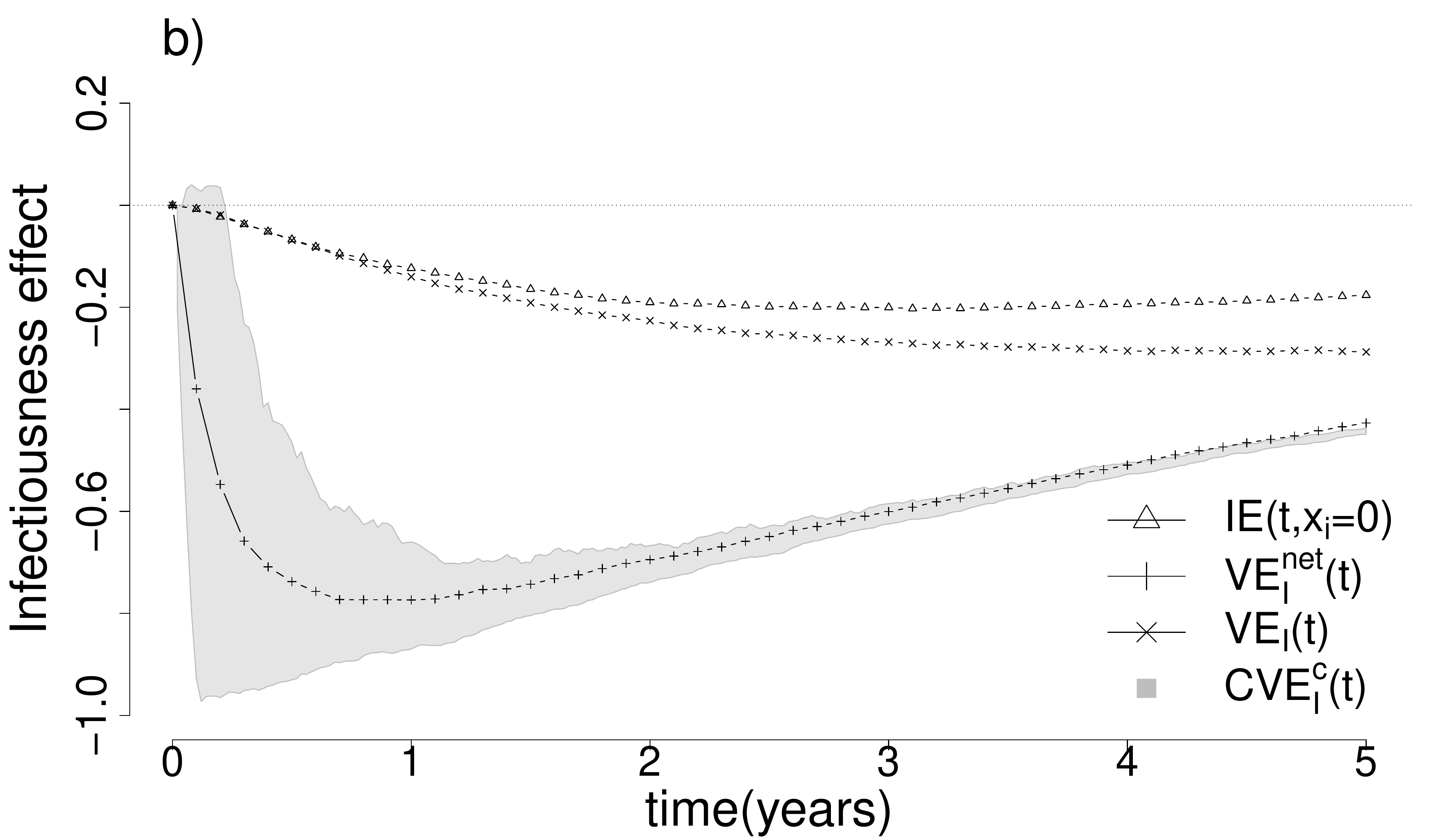}
\end{subfigure}\hfil 
\medskip
\begin{subfigure}{0.49\textwidth}
  \includegraphics[width=\linewidth]{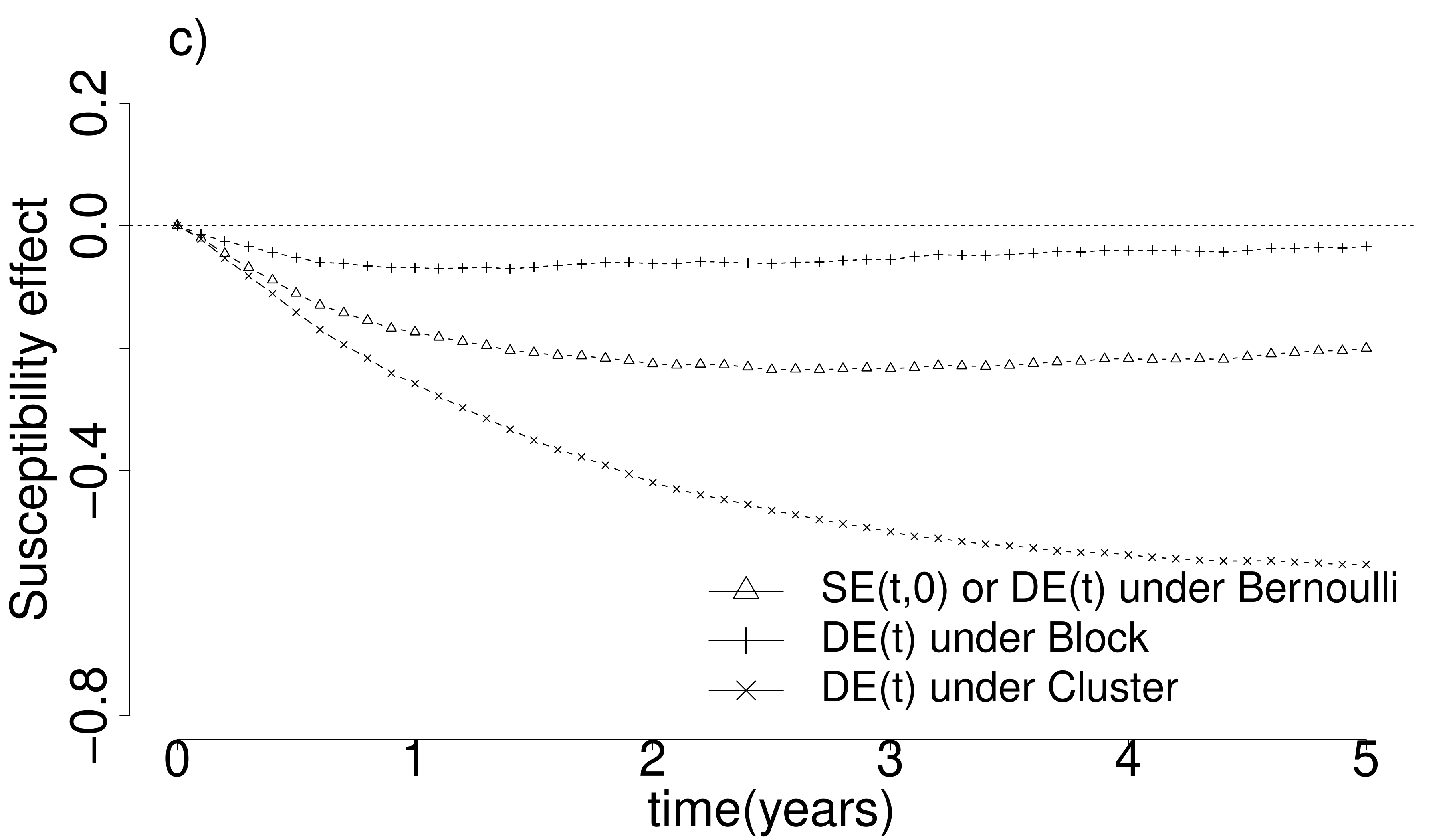}
\end{subfigure}\hfil 
\begin{subfigure}{0.49\textwidth}
  \includegraphics[width=\linewidth]{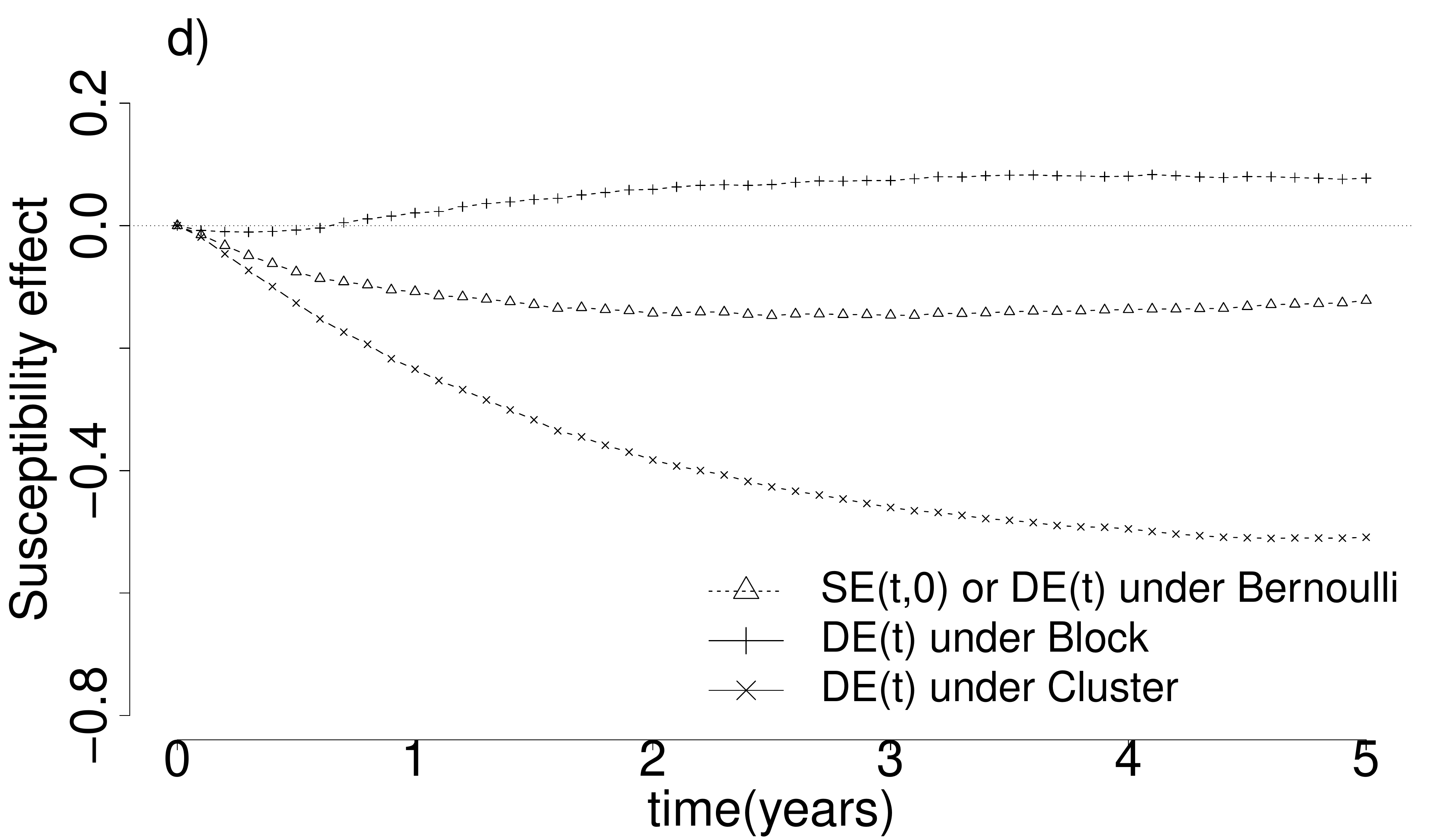}
\end{subfigure}\hfil 
\caption{Comparison of different natural infectiousness and susceptibility effects. Figure a) compares different natural infectiousness effects -- natural infectiousness effect $\IE(t,x_i=0)$, crude infectiousness effect $\VE_{\I}^{\net}(t)$, the infectiousness defined in mediation analysis $\VE_{\I}(t)$ and bounds identified by principal stratification -- when both true susceptibility effect and true infectiousness effect are beneficial ($e^{\beta}=0.2$, $e^{\sigma}=0.5$). Similarly, Figure b) shows the same comparison of multiple natural infectiousness effects as in Figure a) when the true infectiousness effect is much stronger than the true susceptibility effect ($e^{\beta}=0.4$, $e^{\sigma}=0.01$). Figure c) shows the comparison of different types of natural susceptibility effect -- the natural susceptibility effect $\SE(t,0)$, the crude susceptibility effect $\DE(t)$ under Bernoulli, Complete, and Cluster randomization -- when both true susceptibility effect and true infectiousness effect are beneficial ($e^{\beta}=0.2$, $e^{\sigma}=0.5$) as in Figure a). Likewise, Figure d) shows the same comparison of multiple natural susceptibility effects when the true infectiousness effect is much stronger than the true susceptibility effect ($e^{\beta}=0.4$, $e^{\sigma}=0.01$). All four graphs are under constant baseline hazards $\alpha(t)=0.2$ and $\gamma(t)=10$.}
\label{differentestimand}
\end{figure}


\section{Discussion}

We have described a nonparametric framework for identifying causal intervention effects under contagion in general two-person partnerships. The estimands and identification results generalize those given in prior work \citep{vanderweele2012components,vanderweele2011bounding,halloran2012causal,ogburn2017vaccines}, and establish that point identification of clinically meaningful causal estimands under contagion is possible even when relationships are symmetric and either individual can be treated. 
We have made no assumptions about the functional form of infection risks (beyond the independencies and exclusion restrictions implied by Assumptions \ref{as:exclusion}--\ref{as:crossworld}), how the risk of infection to a susceptible individual changes when their partner becomes infected, or how the vaccine changes susceptibility or infectiousness over time.  The framework respects the logic of infectious disease transmission: if the outcome is not transmissible, the contagion and infectiousness effects are zero. 

By studying the role of a partner's infection time in the identification of controlled causal effects, we can identify causal estimands that are both more fundamental and more directly linked to the biological effect of a vaccine on infection risk than simple contrasts of infection rates. Our results also show that while some crude contrasts can recover causal effects in restricted settings (e.g. the infectious effect $\VE_{\I}(t)$ in the asymmetric partnership setting) or under a particular randomization design (e.g. the direct effect $\DE(t)$ under independent Bernoulli randomization), they may not deliver useful summaries of vaccine effects in more general situations.  Finally, the framework developed in this paper may be useful in settings beyond infectious disease epidemiology, where symmetric mediated effects are of interest \citep[e.g.][]{sjolander2016carryover,sherman2018identification}.

One important limitation of our identification approach is that the controlled estimands and cross-world natural estimands require observation of infection times, and not just binary infection indicators at a follow-up time $t$.  In real-world vaccine trials, it may be unreasonable to require investigators to measure infection times $T_i$ with precision, as is required by Lemma \ref{lem:Fi} and Theorem \ref{thm:id}.  Instead, cross-sectional infection assessment, follow-up surveys, or tests for biomarkers of prior infection are commonly used as the primary outcome.  Corollary \ref{cor:idnatural} shows exactly how controlled effects that rely on infection times relate to natural effects that do not.  Attempts to disentangle individual effects from the mediating effects of treatment to partners using only binary infection outcomes may fail to recover useful controlled or marginal effects \citep[see, e.g. $\VE_{\I}^\text{net}$, analyzed by][]{halloran2012causal}.

Finally, while the symmetric partnership setting is useful for conceptualizing, defining, and identifying causal estimands, real-world vaccine trials usually happen in clusters of varying sizes. Adapting the setting outlined here to larger clusters results in rapid expansion of the number of potential outcomes, corresponding to every possible ordering of infections, necessitating simplifying structural assumptions to reduce the dimensionality of the problem.  One promising avenue for dramatically reducing the number of potential outcomes without imposing a parametric structure was proposed by \citet{kenah2013non,kenah2015semi}.  The idea is that contagion works by competing risks, where hazards of infection from different sources are additive.  This approach imposes no additional structure on the distribution of the initial time to infection, but assumes that new infected cluster members always add a competing risk of infection to the already existing risks of infection for susceptibles.  


\section*{Acknowledgements}

WWL was supported by NIH grant R01 AI085073 and by a Gillings Innovation Laboratory award from the UNC Gillings School of Global Public Health.
FWC was supported by NIH grants DP2 OD022614, R01 AI112438, and R01 AI112970.
We are grateful to 
Peter M. Aronow,
Soheil Eshghi,
Eben Kenah,
Olga Morozova,
and
Virginia E. Pitzer,
Li Zeng
for helpful comments and discussion. 




\bibliographystyle{abbrvnat}
\bibliography{contagion_counterfactual}

\appendix
\section{Proofs}

\begin{proof}[Proof of Lemma \ref{lem:Fi}]
  Let $f_i(w|x_i,\bl_i)$ be the density of $W_i(x_i)$ when $\bL_i=\bl_i$ and let $F_i(w|x_i,\bl_i)$ be the corresponding cumulative distribution function.  By Assumption \ref{as:consistency}, $0<F_i(w|x_i,\bl_i)<1$ for all $w>0$, $x_i$, and $\bl_i$, so we can write
  \[ \frac{f_i(w|x_i,\bl_i)}{1-F_i(w|x_i,\bl_i)} = -\frac{d}{dw} \log(1-F_i(w|x_i,\bl_i)) . \]
  Then rearranging, we have 
\begin{equation*}
\begin{split}
   F_i(w|x_i,\bl_i) &= 1 - \exp\left[ -\int_0^w \frac{f_i(u|x_i,\bl_i)}{1-F_i(u|x_i,\bl_i))} du \right] \\
                    &= 1 - \exp\left[ -\int_0^w \frac{f_i(u|x_i,\bl_i)(1-F_j(u|x_j,\bl_j))}{(1-F_i(u|x_i,\bl_i))(1-F_j(u|x_j,\bl_j))} du \right] \\
                    &= 1 - \exp\left[ -\int_0^w \frac{p(W_i(x_i)=u,W_j(x_j)>u|\X=(x_i,x_j),\bL=(\bl_i,\bl_j))}{\Pr(W_i(x_i)>u,W_j(x_j)>u|\X=(x_i,x_j),\bL=(\bl_i,\bl_j))} du \right] \\
                    &= 1 - \exp\left[ -\int_0^w \frac{p(T_i=u,T_j>u|\X=(x_i,x_j),\bL=(\bl_i,\bl_j))}{\Pr(T_i>u,T_j>u|\X=(x_i,x_j),\bL=(\bl_i,\bl_j))} du \right] ,
\end{split}
\end{equation*}
where $x_j$ is any fixed value of $X_j$ and $\bl_j$ is any fixed value of $\bL_j$. 
\end{proof}
\begin{lem}
Under Assumptions \ref{as:exclusion}--\ref{as:txignorability}, $Y_i(t;w_j,\x) \indep W_j(x_j) \mid \bL$ and $Y_i(t;w_j,\x) \indep \X | \bL$. 
\label{lem:outcomeindependence}
\end{lem}
\begin{proof}[Proof of Lemma \ref{lem:outcomeindependence}]
Fix a value $w_j>0$ and let $\x=(x_i,x_j)$.  If $W_i(x_i) < w_j$, then $T_i(w_j,\x) = W_i(x_i)$ and by Assumption \ref{as:exclusion}, $W_i(x_i) \indep W_j(x_j) \mid \bL$, so $T_i(w_j,\x) \indep W_j(x_j) \mid \bL$.  If $W_i(x_i)>w_j$ then $T_i(w_j,\x) = w_j+Z_i(w_j,\x)$ and by Assumption \ref{as:infignorability} $Z_i(w_j,\x) \indep W_j \mid \bL$, so $T_i(w_j,\x) \indep W_j(x_j) \mid \bL$. Therefore, since $Y_i(t;w_j,\x) = \indicator{T_i(w_j,\x)<t}$, it follows that $Y_i(t;w_j,\x) \indep W_j(x_j) \mid \bL$. 
 
By the same reasoning, if $W_i(x_i) < w_j$, then $T_i(w_j,\x) = W_i(x_i)$ and by Assumption \ref{as:txignorability}, $W_i(x_i) \indep \X \mid \bL$.  If $W_i(x_i)>w_j$ then $T_i(w_j,\x) = w_j+Z_i(w_j,\x)$ and by Assumption \ref{as:txignorability}, $Z_i(w_j,\x) \indep X \mid \bL$. Therefore, since $Y_i(t;w_j,\x) = \indicator{T_i(w_j,\x)<t}$, it follows that $Y_i(t;w_j,\x) \indep \X \mid \bL$. 
\end{proof}
\begin{lem}
Under Assumptions \ref{as:exclusion}-\ref{as:consistency}, $\E[Y_i(t,w_j,\x)]=\E[Y_i(t)|W_j=w_j,\X=\x]$.
\label{le:consistency}
\end{lem}
\begin{proof}[Proof of Lemma \ref{le:consistency}]
Fix a value $w_j>0$ and $\x=(x_i,x_j)$. If $W_i(x_i) \ge w_j$ then
\begin{equation*}
\begin{split}
\E[Y_i(t,w_j,\x)] &= \Pr(T_i(w_j,\x) <t)  \text{ by the definition of $Y_i(t,w_j,\x)$} \\
    &= \Pr(w_j+Z_i(w_j,\x)<t) \text{ by the definition of $T_i(w_j,\x)$ and $W_i(x_i) \ge w_j$} \\
    &= \Pr(Z_i(w_j,\x) < t-w_j) \\
    &= \Pr(Z_i(w_j,\x) < t-w_j | W_j=w_j) \text{ by Assumption \ref{as:infignorability}} \\
    &= \Pr(Z_i(w_j,\x) < t-w_j | W_j=w_j, \X=\x) \text{ by Assumption \ref{as:txignorability}} \\
    &= \Pr(Z_i < t-w_j | W_j=w_j, \X=\x) \text{ by Assumption \ref{as:consistency}}\\
    &= \Pr(Z_i < t-W_j | W_j=w_j, \X=\x) \\
    &= \Pr(Z_i +W_j < t | W_j=w_j, \X=\x) \\
    &= \Pr(T_i <t| W_j=w_j, \X=\x) \text{ by the definition of $T_i$}\\
    &= \E[Y_i(t)| W_j=w_j, \X=\x] \text{ by the definition of $Y_i(t)$}
\end{split}
\end{equation*}
If $W_i(x_i)<w_j$ then
\begin{equation*}
\begin{split}
\E[Y_i(t,w_j,\x)] &= \Pr(T_i(w_j,\x) <t)  \text{ by the definition of $Y_i(t,w_j,\x)$} \\
    &= \Pr(W_i(x_i)<t) \text{ by the definition of $T_i(w_j,\x)$ and $W_i(x_i)<w_j$} \\
    &= \Pr(W_i(x_i)<t|X_i=x_i,X_j=x_j) \text{ by Assumption \ref{as:txignorability}}\\
    &= \Pr(W_i(x_i)<t|W_j=w_j,X_i=x_i,X_j=x_j) \text{ by Assumption \ref{as:exclusion}}\\
    &= \Pr(W_i<t|W_j=w_j,\X=\x) \text{ by Assumption \ref{as:consistency}}\\
    &= \Pr(T_i <t| W_j=w_j, \X=\x) \text{ by the definition of $T_i$}\\
    &= \E[Y_i(t)| W_j=w_j, \X=\x] \text{ by the definition of $Y_i(t)$}
\end{split}
\end{equation*}
\end{proof}
\begin{proof}[Proof of Theorem \ref{thm:id}]
The average potential infection outcome when $\bL=\bl$ is given by 
\begin{equation*}
\begin{split}
    & \E[Y_i(t;w_j,\x)|\bL=\bl]= \E[Y_i(t;w_j,\x)|W_j=w_j,\X=\x,\bL=\bl] \\
    \intertext{by Lemma \ref{lem:outcomeindependence}}
    &= \E[Y_i(t;w_j,\x)|W_i\le w_j,W_j=w_j,\X=\x,\bL=\bl] \Pr(W_i\le w_j|W_j=w_j,\X=\x,\bL=\bl]) \\
    &\quad + \E[Y_i(t;w_j,\x)|W_i>w_j,W_j=w_j,\X=\x,\bL=\bl] \Pr(W_i>w_j|W_j=w_j,\X=\x,\bL=\bl]) \\
    &= \Pr(T_i(w_j,\x)<t|W_i\le w_j,W_j=w_j,\X=\x,\bL=\bl) \Pr(W_i\le w_j|W_j=w_j,\X=\x,\bL=\bl]) \\
    &\quad + \E[Y_i(t;w_j,\x)|W_i>w_j,W_j=w_j,\X=\x,\bL=\bl] \Pr(W_i> w_j|W_j=w_j,\X=\x,\bL=\bl]) \\
    \intertext{by the definition of $Y_i(t;w_j,\x)$}
    &= \Pr(W_i(x_i)<t|W_i\le w_j,W_j=w_j,\X=\x,\bL=\bl) \Pr(W_i\le w_j|W_j=w_j,\X=\x,\bL=\bl]) \\
    &\quad + \E[Y_i(t;w_j,\x)|W_i>w_j,W_j=w_j,\X=\x,\bL=\bl] \Pr(W_i> w_j|W_j=w_j,\X=\x,\bL=\bl]) \\
    \intertext{by the definition of $T_i(w_j,\x)$}
    &= \Pr(W_i(x_i)<t|W_i\le w_j,X_i=x_i,\bL=\bl) \Pr(W_i\le w_j|X_i=x_i,\bL=\bl]) \\
    &\quad + \E[Y_i(t;w_j,\x)|W_i>w_j,W_j=w_j,\X=\x,\bL=\bl] \Pr(W_i> w_j|X_i=x_i,\bL=\bl]) \\
    \intertext{by Assumption \ref{as:exclusion}}
    &= \Pr(W_i<t|W_i\le w_j,X_i=x_i,\bL=\bl) \Pr(W_i\le w_j|X_i=x_i,\bL=\bl]) \\
    &\quad + \E[Y_i(t)|W_i>w_j,W_j=w_j,\X=\x,\bL=\bl] \Pr(W_i> w_j|X_i=x_i,\bL=\bl]) \\
    \intertext{by Assumption \ref{as:consistency} and Lemma \ref{le:consistency}}
    &= \Pr(W_i<t, W_j \le w_j|X_i=x_i,\bL=\bl) \\
    &\quad + \E[Y_i(t)|W_i>w_j,W_j=w_j,\X=\x,\bL=\bl] \Pr(W_i> w_j|X_i=x_i,\bL=\bl]) \\
\end{split}
\label{maineq:the1}
\end{equation*}
When $t \ge w_j$, then    
\begin{equation*}
\begin{split}
\E[Y_i(t;w_j,\x)|\bL=\bl]  &= \Pr(W_i<t, W_j \le w_j|X_i=x_i,\bL=\bl) \\
    &\quad + \E[Y_i(t)|W_i>w_j,W_j=w_j,\X=\x,\bL=\bl] \Pr(W_i> w_j|X_i=x_i,\bL=\bl]) \\
    &= \Pr(W_i\le w_j|X_i=x_i,\bL=\bl]) \\
    &\quad + \E[Y_i(t)|W_i>w_j,W_j=w_j,\X=\x,\bL=\bl] \Pr(W_i> w_j|X_i=x_i,\bL=\bl]) \\
    &= F_i(w_j|x_i,\bl_i) + (1-F_i(w_j|x_i,\bl_i))\E[Y_i(t)|W_i>w_j,W_j=w_j,\X=\x,\bL=\bl].
\end{split}
\end{equation*}
Likewise, when $t <w_j$, then 
\begin{equation*}
\begin{split}
\E[Y_i(t;w_j,\x)|\bL=\bl] &= \Pr(W_i<t, W_j \le w_j|X_i=x_i,\bL=\bl) \\
    &\quad + \E[Y_i(t)|W_i>w_j,W_j=w_j,\X=\x,\bL=\bl] \Pr(W_i> w_j|X_i=x_i,\bL=\bl]) \\
    &= \Pr(W_i < t|X_i=x_i,\bL_i=\bl_i])  \\
    &\quad + \E[Y_i(t)|W_i>w_j,W_j=w_j,\X=\x,\bL=\bl] \Pr(W_i>w_j|X_i=x_i,\bL_i=\bl_i]) \\
    &= \Pr(W_i \le t|X_i=x_i,\bL_i=\bl_i])  \\
    \intertext{since $\E[Y_i(t)|W_i>w_j,W_j=w_j,\X=\x,\bL=\bl]=0$ when $t<w_j$}
    &= F_i(w_j|x_i,\bl_i) .
\end{split}
\end{equation*}
\end{proof}
\begin{proof}[Proof of Corollary \ref{cor:idnatural}]
\begin{equation*}
\begin{split}
  \E[Y_i(t;W_j(x_j),\x)|\bL=\bl] &= \E \big{[} \E[Y_i(t;W_j(x_j),\x)|\bL=\bl]  \big{]}  \\
   &= \int_0^\infty \E[Y_i(t;u,\x)|W_j=u,\X=\x,\bL=\bl] dF_j(u|x_j,\bl_i) \text{ by Assumption \ref{as:exclusion}}\\
   &= \int_0^\infty \E[Y_i(t)|W_j=u,\X=\x,\bL=\bl] dF_j(u|x_j,\bl_i) \text{ by Lemma \ref{le:consistency} and Assumption \ref{as:consistency}} \\
   &= \E[Y_i(t)|\X=\x,\bL=\bl]. \\
\end{split}
\end{equation*}
Likewise, when $\x=(x_i,x_j)$ and $x_j'\neq x_j$, 
\begin{equation*}
\begin{split}
  \E[Y_i(t;W_j(x_j'),\x|\bL=\bl] &= \E \big{[} \E[Y_i(t;W_j(x_j),\x|\bL=\bl]  \big{]} \\
   &= \int_0^\infty \E[Y_i(t;u,\x)|W_j=u,\X=\x,\bL=\bl] dF_j(u|x'_j,\bl_i) \text{ by Assumption \ref{as:exclusion}} \\
   &= \int_0^\infty \E[Y_i(t)|W_j=u,\X=\x,\bL=\bl] dF_j(u|x'_j,\bl_i) \text{ by Lemma \ref{le:consistency} and Assumption \ref{as:consistency}} \\
\end{split}
\end{equation*}
\end{proof}
\begin{lem}
When $\SE(t,w_j,x_j)=0$, then $F_j(t|x_j)=F_j(t|1-x_j)$ and $\E[Y_i(t)|X_i=1,X_j=x_j]=\E[Y_i(t)|X_i=0,X_j=x_j]$, for all $x_j \in \{0,1\}$ and $t \ge 0$. 

When $\SE(t,w_j,x_j)=\IE(t,w_j,x_i)=0$, then $\E[Y_i(t)|X_i=0,X_j=1] = \E[Y_i(t)|X_i=0,X_j=0]$.

When $\SE(t,w_j,x_j)=0$ and $\IE(t,w_j,x_i)<0$, then $\E[Y_i(t)|X_i=0,X_j=1] < \E[Y_i(t)|X_i=0,X_j=0]$.
\label{lem:fortheorem3_1}
\end{lem}
\begin{proof}[Proof of Lemma \ref{lem:fortheorem3_1}]
First we prove $F_j(t|x_j)=F_j(t|1-x_j)$, for all $x_j \in \{0,1\}$ when $\SE(t,w_j,x_j)=0$.
\begin{equation}
\begin{split}
F_j(t|x_j) &=\Pr(W_j(x_j)<t) = \Pr(T_j(w_i=\infty,x_i,x_j)<t) \text{ by the definition of $T_j(w_i,x_j,x_i)$}\\
    &=\E[Y_j(t;w_i=\infty,x_j,x_i)] \text{ by the definition of $Y_j(u;w_i,x_j,x_i)$}
\\
    &=\E[Y_i(t;w_i=\infty,x'_j,x_i)] \text{ since $\SE(t,w_j,x_j)=0$} \\
    &= \Pr(W_j(x'_j)<t)=F_j(t|x'_j).
\end{split}
\label{eq:seequal1}
\end{equation}
Second, we prove $\E[Y_i(t)|X_i=1,X_j=x_j]=\E[Y_i(t)|X_i=0,X_j=x_j]$ for all $x_j \in \{0,1\}$, if $\SE(t,w_j,x_j)=0$.
\begin{equation}
\begin{split}
\E[Y_i(t)|X_i=1,X_j=x_j] &= \int_0^\infty \E[Y_i(t)|W_j=u,X_i=1,X_j=x_j] dF_j(u|x_j) \text{ by Assumption \ref{as:exclusion}}\\
    &= \int_0^\infty \E[Y_i(t;u,x_i=1,x_j)] dF_j(u|x_j) \text{ by Lemma \ref{le:consistency}}\\
    &= \int_0^\infty \E[Y_i(t;u,x_i=0,x_j)] dF_j(u|x_j) \text{ since $\SE(t,w_j,x_j)=0$} \\
    &= \E[Y_i(t)|X_i=0,X_j=x_j].
\end{split}
\label{eq:seequal0}
\end{equation}

Third, by (\ref{eq:seequal1}), we prove $\E[Y_i(t)|X_i=0,X_j=1] = \E[Y_i(t)|X_i=0,X_j=0]$, if $\SE(t,w_j,x_j)=\IE(t,w_j,x_i)=0$.
\begin{equation}
\begin{split}
\E[Y_i(t)|X_i=0,X_j=1] &= \int_0^t \E[Y_i(t)|W_j=u,X_i=0,X_j=1] dF_j(u|1) \text{ by Assumption \ref{as:exclusion}}\\
    &= \int_0^\infty \E[Y_i(t;u,x_i=0,x_j=1)] dF_j(u|1) \text{ by Lemma \ref{le:consistency}} \\ 
    &= \int_0^\infty \E[Y_i(t;u,x_i=0,x_j=0)] dF_j(u|1) \text{ since $\IE(t,w_j,x_i)=0$} \\
    &= \int_0^\infty \E[Y_i(t;u,x_i=0,x_j=0)] dF_j(u|0) \text{ by (\ref{eq:seequal1})}\\
    &= \int_0^\infty \E[Y_i(t)|W_j=u,X_i=0,x_j=0)] dF_j(u|0) \text{ by Lemma \ref{le:consistency}}\\
    &= \E[Y_i(t)|X_i=0,X_j=0].
\end{split}
\label{eq:se}
\end{equation} 

Fourth, by (\ref{eq:seequal1}), we prove $\E[Y_i(t)|X_i=0,X_j=1] < \E[Y_i(t)|X_i=0,X_j=0]$, if $\SE(t,w_j,x_j)=0$ and $\IE(t,w_j,x_i)<0$.
\begin{equation}
\begin{split}
\E[Y_i(t)|X_i=0,X_j=1] &= \int_0^t \E[Y_i(t)|W_j=u,X_i=0,X_j=1] dF_j(u|1) \text{ by Assumption \ref{as:exclusion}}\\
    &= \int_0^\infty \E[Y_i(t;u,x_i=0,x_j=1)] dF_j(u|1) \text{ by Lemma \ref{le:consistency}} \\ 
    &< \int_0^\infty \E[Y_i(t;u,x_i=0,x_j=0)] dF_j(u|1) \text{ since $\IE(t,w_j,x_i)<0$} \\
    &= \int_0^\infty \E[Y_i(t;u,x_i=0,x_j=0)] dF_j(u|0) \text{ by (\ref{eq:seequal1})}\\
    &= \int_0^\infty \E[Y_i(t)|W_j=u,X_i=0,x_j=0)] dF_j(u|0) \text{ by Lemma \ref{le:consistency}}\\
    &= \E[Y_i(t)|X_i=0,X_j=0].
\end{split}
\label{ieq:se}
\end{equation}
\end{proof}

\begin{proof}[Proof of Theorem \ref{thm:VEAR}]

Given the conclusions from (\ref{eq:seequal0}) and (\ref{ieq:se}), we have
\begin{equation}
\begin{split}
\DE(t) &= \E[Y_i(t)|X_i=1] - \E[Y_i(t)|X_i=0] \\
      &= \E[Y_i(t)|X_i=1,X_j=1] \Pr(X_j=1|X_i=1) + \E[Y_i(t)|X_i=1,X_j=0] \Pr(X_j=0|X_i=1) \\
      & \quad -\E[Y_i(t)|X_i=0,X_j=1] \Pr(X_j=1|X_i=0)- \E[Y_i(t)|X_i=0,X_j=0] \Pr(X_j=0|X_i=0) \\
      &= \E[Y_i(t)|X_i=0,X_j=1] \Pr(X_j=1|X_i=1) + \E[Y_i(t)|X_i=0,X_j=0] \Pr(X_j=0|X_i=1) \\
      & \quad -\E[Y_i(t)|X_i=0,X_j=1] \Pr(X_j=1|X_i=0)- \E[Y_i(t)|X_i=0,X_j=0] \Pr(X_j=0|X_i=0) \\
      \intertext{by (\ref{eq:seequal0}) in Lemma \ref{lem:fortheorem3_1}
}
      &= \E[Y_i(t)|X_i=0,X_j=1] \Big{[} \Pr(X_j=1|X_i=1) - \Pr(X_j=1|X_i=0)\Big{]}\\
      & \quad  + \E[Y_i(t)|X_i=0,X_j=0] \Big{[} \Pr(X_j=0|X_i=1) - \Pr(X_j=0|X_i=0)\Big{]}\\
      &= \E[Y_i(t)|X_i=0,X_j=1] \Big{[} \Pr(X_j=1|X_i=1) - \Pr(X_j=1|X_i=0)\big{]}\\
      &\quad + \E[Y_i(t)|X_i=0,X_j=0] \Big{\{} [1-\Pr(X_j=1|X_i=1)] - [1-\Pr(X_j=1|X_i=0)]\Big{\}}\\
      &= \Big{\{} \E[Y_i(t)|X_i=0,X_j=1]- \E[Y_i(t)|X_i=0,X_j=0] \Big{\}} \cdot  \Big{[} \Pr(X_j=1|X_i=1) - \Pr(X_j=1|X_i=0)\Big{]}
\end{split}
\label{mainproof3}
\end{equation}

Note by (\ref{ieq:se}) in Lemma \ref{lem:fortheorem3_1}, we have the first term at the last line of (\ref{mainproof3}) being negative. The sign of $\DE(t)$ then depends only on the treatment assignment mechanism, which leads to the following conclusions for $\DE(t)$.

1. If the treatment assignment is positively correlated ($\Pr(X_i=c,X_j=c) > \Pr(X_i=c) \Pr(X_j=c)$ for $c\in{0,1}$), we have:
\begin{equation}
\begin{split}
& \Pr(X_j=1|X_i=1) - \Pr(X_j=1|X_i=0) \\
&= \frac{\Pr(X_j=1,X_i=1)}{\Pr(X_i=1)}- \frac{\Pr(X_j=1,X_i=0)}{\Pr(X_i=0)} \\
&=\frac{\Pr(X_j=1,X_i=1)\Pr(X_i=0)-\Pr(X_j=1,X_i=0)\Pr(X_i=1)}{\Pr(X_i=1)\Pr(X_i=0)} \\
&=\frac{\Pr(X_j=1,X_i=1)[1-\Pr(X_i=1)]-\Pr(X_j=1,X_i=0)\Pr(X_i=1)}{\Pr(X_i=1)\Pr(X_i=0)} \\
&=\frac{\Pr(X_j=1,X_i=1)-\Pr(X_j=1,X_i=1)\Pr(X_i=1)-\Pr(X_j=1,X_i=0)\Pr(X_i=1)}{\Pr(X_i=1)\Pr(X_i=0)} \\
&=\frac{\Pr(X_j=1,X_i=1)-\Pr(X_j=1\Pr(X_i=1)}{\Pr(X_i=1)\Pr(X_i=0)} \ge 0 \\
\end{split}
\label{treatmentassignment}
\end{equation}
Thus, $\DE(t)<0$.

2. If the treatment assignment is independent ($\Pr(X_i=c,X_j=c) = \Pr(X_i=c) \Pr(X_j=c)$ for $c\in{0,1}$), then by similar arguments of (\ref{treatmentassignment}), we have $\Pr(X_j=1|X_i=1) -\Pr(X_j=1|X_i=0) =0$. Thus, $\DE(t)=0$.

3. If the treatment assignment is negatively correlated ($\Pr(X_i=c,X_j=c) < \Pr(X_i=c) \Pr(X_j=c)$ for $c\in{0,1}$), then  by similar arguments of (\ref{treatmentassignment}), we have $\Pr(X_j=1|X_i=1) -\Pr(X_j=1|X_i=0) <0$. Thus, $\DE(t)>0$.

When $\IE(t,w_j,x_i)=0$, following (\ref{eq:se}) and (\ref{mainproof3}) in Lemma \ref{lem:fortheorem3_1}, we have $\E[Y_i(t)|X_i=0,X_j=1]=\E[Y_i(t)|X_i=0,X_j=0]$ and thus $\DE(t)=0$.

Similar arguments apply for $\VE_{\AR}(t)$.
\end{proof}
\begin{proof}[Proof of Theorem \ref{thm:VEInet}]
We evaluate the sign of $\VE_I^{net}(t)$ by analyzing $\SAR_{00}(t)-\SAR_{10}(t)$.
\[
\VE_I^{net}(t)=1-\frac{\SAR_{10}(t)}{\SAR_{00}(t)}=\frac{\SAR_{00}(t)-\SAR_{10}(t)}{\SAR_{00}(t)}
\]
First, we analyze the sign of $\VE_I^{net}(t)$ under a null true infectiousness effect, when the infection outcome is positively contagions and vaccine has a favorable effect prior first infection through $h_0(u|1)=\varepsilon h_0(u|0)$, for $\varepsilon \in [0,1)$.
\begin{equation}
\begin{split}
    & \SAR_{10}(t) - \SAR_{00}(t) \\
    &= \E[Y_i(t)|T_j<t,T_i>T_j,X_i=0,X_j=1] - \E[Y_i(t)|T_j<t,T_i>T_j,X_i=0,X_j=0] \\
    &= \frac{\int_0^t \E[Y_i(t)|W_j=u,W_i>u,\X=(0,1)](1-F_i(u|0)) dF_j(u|1)}{\Pr(W_j<t,W_i>W_j|\X=(0,1))}  \\
    & \qquad  - \frac{\int_0^t \E[Y_i(t)|W_j=u,W_i>u,\X=(0,0)](1-F_i(u|0))dF_j(u|0)}{\Pr(W_j<t,W_i>W_j|\X=(0,0))}  \\
\intertext{by applying the law of total probability}
    &= \int_0^t \E[Y_i(t)|W_j=u,W_i>u,\X=(0,1)] \frac{(1-F_i(u|0))dF_j(u|1)}{\int_0^t (1-F_i(v|0))dF_j(v|1)}  \\
    & \qquad  - \int_0^t \E[Y_i(t)|W_j=u,W_i>u,\X=(0,0)] \frac{(1-F_i(u|0))dF_j(u|0)}{\int_0^t (1-F_i(v|0))dF_j(v|0)} \\
    &= \int_0^t  \E[Y_i(t)|W_j=u,W_i>u,\X=(0,0)] \Big{[} \frac{(1-F_i(u|0))dF_j(u|1)}{\int_0^t (1-F_i(v|0))dF_j(v|1)} - \frac{(1-F_i(u|0))dF_j(u|0)}{\int_0^t (1-F_i(v|0))dF_j(v|0)}  \Big{]}. \\
\intertext{By $\IE(t,w_j,0)=0$ and Lemma \ref{le:consistency}}\end{split}
\label{SARDIFF}
\end{equation}
To ease the notation in Equation (\ref{SARDIFF}), we denote $\E[Y_i(t)|W_j=u,W_i>u,\X=(0,0)]=k(u)$. Denote $g(u|1)=\frac{(1-F_i(u|0))dF_j(u|1)}{\int_0^t (1-F_i(v|0))dF_j(v|1)}$ and $g(u|0)=\frac{(1-F_i(u|0))dF_j(u|0)}{\int_0^t (1-F_i(v|0))dF_j(v|0)}$, and $G(u|1)=\int_0^u g(s|1)ds$ and $G(u|0)=\int_0^u g(s|0)ds$. Then by integration by parts, 
(\ref{SARDIFF}) can be re-written as follows:
\begin{equation*}
\begin{split}
\SAR_{10}(t) - \SAR_{00}(t) &= \int_0^t k(u)[g(u|1)-g(u|0)] du \\
    & = k(u)[G(u|1)-G(u|0)] \Big|_0^t - \int_0^t (G(u|1)-G(u|0)) dk(u).\\
\end{split}
\label{decompositionSAR}
\end{equation*}
By their definitions, we have $G(0|1)-G(0|0)=0$ and $G(t|1)-G(t|0)=0$, and thus $k(u)[G(u|1)-G(u|0)] \Big|_0^t =0$. In other words, the sign of $\SAR_{10}(t) - \SAR_{00}(t)$ only depends on the sign of $G(u|1)-G(u|0)$ and $dk(u)$ for all $u>0$. First, we can show that $dk(u)<0$ for $0 \le u <t$. For $0 \le u <u' <t$, we have
\begin{equation}
\begin{split}
k(u) &= \frac{ \E[Y_i(t)|W_j=u,\X=(0,0)] - F_i(u|0)}{1-F_i(u|0)} \text{ by Theorem \ref{thm:id} }\\
&> \frac{ \E[Y_i(t)|W_j=u',\X=(0,0)] - F_i(u|0)}{1-F_i(u|0)} \text{ by $\CE(t,u,u',(0,0))>0$}\\
&= \frac{ \E[Y_i(t)|W_j=u',\X=(0,0)] - F_i(u'|0)+F_i(u'|0)-F_i(u|0)}{1-F_i(u|0)}\\
&=\frac{ k(u') (1-F_i(u'|0))+F_i(u'|0)-F_i(u|0)}{1-F_i(u|0)} \text{ by Theorem \ref{thm:id} }\\
& \ge \frac{ k(u') (1-F_i(u'|0))+(F_i(u'|0)-F_i(u|0)) k(u')}{1-F_i(u|0)} \text{ by $k(u')\le 1$}\\
& = \frac{ k(u') (1-F_i(u|0))}{1-F_i(u|0)}=k(u').
\end{split}
\label{decreasingk}
\end{equation}

Next, we analyze the property of $G(u|1)-G(u|0)$ for $\forall u>0$. Denote $H_0(u)=\int_0^u h_0(s|0)ds$. Given $h_0(u|1)=\varepsilon h_0(u|0)$ with $\varepsilon \in [0,1)$, we can write out $G(u|0)$ and $G(u|1)$ in terms of $h_0(u|0)$ as follows.
\begin{equation}
\begin{split}
G(s|1) &=\frac{\int_0^s (1-F_i(u|0)) dF_j(u|1) }{\int_0^t(1-F_i(v|0))dF_j(v|1)} = \frac{\int_0^s \varepsilon\cdot h_0(u|0) e^{-\varepsilon\cdot H_0(u)}e^{-H_0(u)} du}{\int_0^t \varepsilon\cdot h_0(v|0) e^{-\varepsilon\cdot H_0(v)}e^{-H_0(v)} dv} \\
    & =  \frac{\int_0^s \varepsilon\cdot h_0(u|0) e^{-(\varepsilon+1)\cdot H_0(u)} du}{\int_0^t \varepsilon\cdot h_0(v|0) e^{-(\varepsilon+1)\cdot H_0(v)} dv} = \frac{1-e^{-(\varepsilon+1)H_0(s)}}{1-e^{-(\varepsilon+1)H_0(t)}}\\
G(s|0) &= \frac{\int_0^s (1-F_i(u|0)) dF_j(u|1) } {\int_0^t (1-F_i(v|0))d F_j(v|1)} =  \frac{1-e^{-2H_0(s)}}{1-e^{-2H_0(t)}}\\
\end{split}
\label{dominance1}
\end{equation}
From (\ref{dominance1}), we observe that $G(s|1)$ and $G(s|0)$ only differ by the terms in front of $H_0$. Treat $G(s|1)$ and $G(s|0)$ as functions of $\varepsilon$, and we can re-express them as $G(\varepsilon)=\frac{1-e^{-(\varepsilon+1)H_0(s)}}{1-e^{-(\varepsilon+1)H_0(t)}}$ and $G(1)=\frac{1-e^{-2H_0(s)}}{1-e^{-2H_0(t)}}$, given $\varepsilon<1$. Then, if $G(\varepsilon)$ is a decreasing function of $\varepsilon$, we have $G(u|1) - G(u|0) \le 0$.
\begin{equation}
\frac{\partial}{\partial \varepsilon}G(\varepsilon) = \frac{H_0(u)e^{-(\varepsilon+1)H_0(u)}[1-e^{-(\varepsilon+1)H_0(t)}]-H_0(t)e^{-(\varepsilon+1)H_0(t)}[1-e^{-(\varepsilon+1)H_0(u)}]}{[1-e^{-(\varepsilon+1)H_0(t)}]^2}
\label{Gderivative}
\end{equation}
Divide the numerator of (\ref{Gderivative}) by a positive constant $H_0(t)H_0(u)e^{-(\varepsilon+1) [H_0(u)+H_0(t)]}$. We then have if $\frac{e^{(\varepsilon+1)H_0(u)}-1}{H_0(u)} 
\le \frac{e^{(\varepsilon+1)H_0(t)}-1}{H_0(t)}$ for $u<t$, then $G(u|1) - G(u|0) \le 0$. Treat $\frac{e^{(\varepsilon+1)H_0(t)}-1}{H_0(t)}$ as a function of $u$, given $0 \le u < t$. We have,
\begin{equation}
\begin{split}
\frac{\partial}{\partial u} \frac{e^{(\varepsilon+1)H_0(u)}-1}{H_0(u)} &= \frac{ (\varepsilon+1)H_0(u)e^{(\varepsilon+1)H_0(u)}-e^{(\varepsilon+1)H_0(u)}+1}{[H_0(u)]^2} \\
&= \frac{ (\varepsilon+1)H_0(u)-1+e^{-(\varepsilon+1)H_0(u)}}{[H_0(u)]^2 e^{(\varepsilon+1)H_0(u)}} \text{ by $e^{-(\varepsilon+1)H_0(u)} \ge 1-(\varepsilon+1)H_0(u)$ }
 \\
& \ge 0.
\end{split}
\label{Hderivative}
\end{equation}
Combining (\ref{Gderivative}) and (\ref{Hderivative}), we have $G(u|1) - G(u|0) \le 0$.

In summary, we can see that 
\begin{equation*}
\SAR_{10}(t) - \SAR_{00}(t) = k(u)[G(u|1)-G(u|0)] \Big|_0^t - \int_0^t (G(u|1)-G(u|0)) dk(u) <0
\end{equation*}

Thus, $\VE_I^{net}(t)=1-\frac{\SAR_{10}(t)}{\SAR_{00}(t)}=\frac{\SAR_{00}(t)-\SAR_{10}(t)}{\SAR_{00}(t)}>0$.

Second, we analyze the sign of $\VE_I^{net}(t)$ under a null true infectiousness effect, when the true susceptibility effect is also null.
\begin{equation*}
\begin{split}
& \SAR_{10}(t) - \SAR_{00}(t) \\
&= \int_0^t  \E[Y_i(t)|W_j=u,W_i>u,\X=(0,0)] \Big{[} \frac{(1-F_i(u|0))dF_j(u|1)}{\int_0^t (1-F_i(v|0))dF_j(v|1)} - \frac{(1-F_i(u|0))dF_j(u|0)}{\int_0^t (1-F_i(v|0))dF_j(v|0)}  \Big{]} \\
\intertext{by (\ref{SARDIFF})}
    &= \int_0^t  \E[Y_i(t)|W_j=u,W_i>u,\X=(0,0)] \Big{[} \frac{(1-F_i(u|0))dF_j(u|0)}{\int_0^t (1-F_i(v|0))dF_j(v|0)} - \frac{(1-F_i(u|0))dF_j(u|0)}{\int_0^t (1-F_i(v|0))dF_j(v|0)}  \Big{]} \\
   \intertext{by $\SE(t,w_j,x_j) = 0$ and (\ref{eq:seequal1})}
   &=0.
\end{split}
\end{equation*}
Thus, $\VE_I^{net}(t)=1-\frac{\SAR_{10}(t)}{\SAR_{00}(t)}=\frac{\SAR_{00}(t)-\SAR_{10}(t)}{\SAR_{00}(t)}=0$.

Third, we analyze the sign of $\VE_I^{net}(t)$ in the case of no contagion, when the true susceptibility effect is beneficial. First, $\CE(t,w_j,w'_j,\mathbf{0})=0$ for all $0 <w_j <w'_j$ implies $\IE(t,w_j,0)=0$.
\begin{equation*}
\begin{split}
\IE(t,w_j,x_i) &= \E[ Y_i(t;w_j,x_i,x_j=1) - Y_i(t;w_j,x_i,x_j=0)]  \\
&= \E[ \indicator{W_i(x_i)<t} - \indicator{W_i(t;x_i)<t}]=0
\end{split}
\end{equation*}
Following the same proof for the first case except replacing the second line of (\ref{decreasingk}) by an equal sign, we know $\VE_I^{net}(t)>0$.
\end{proof}
\begin{proof}[Proof of Theorem \ref{th:asy_infectiousness}]
Given $h^i_0(t|0)=0$, we have $F_i(s|0)=1-e^{-\int_0^s h^i_0(u|0)du}=0$ for $W_i(0)$.
\begin{equation}
\begin{split}
\E[Y_i(t;Y_j(x'_j),(0,x_j))|h^i_0(t|0)=0] &= \E[Y_i(t;Y_j(x'_j),(0,x_j))|W_i(0)=\infty]\\
\intertext{ by $F_i(s|x_i)=0$ for $\forall s>0$ }
&=\E[Y_i(t;\mathbf{1}\{W_j(x'_j)<t\},(0,x_j))|W_i(x_i)=\infty] \\
\intertext{given $Y_i(x'_j)=\mathbf\{T_j(x'_j)<t\}$ and $T_j(x'_j)=W_j(0)$ when $W_i(0)=\infty$}
&= \E[Y_i(t;W_j(x'_j),(0,x_j))|W_i(0)=\infty] \\
&= \E[Y_i(t;W_j(x'_j),(0,x_j))|h^i_0(s|0)=0] \\
\end{split}
\label{eq:replace}
\end{equation}
Thus, by the definition of $\VE_I(t)$ and $\IE(t,x_i)$, we have:
\begin{equation*}
\begin{split}
\VE_I(t) &= \E[Y_i(t;Y_j(1),(0,1))-Y_i(t;Y_j(1),(0,0))|h^i_0(s|0)=0]\\
&= \E[Y_i(t;W_j(1),(0,1))-Y_i(t;W_j(1),(0,0))|h^i_0(s|0)=0] = \IE(t,0|h^i_0(s|0)=0)
\end{split}
\end{equation*}
Thus, $\VE_I$ is equivalent to the natural infectiousness effect under the asymmetric partnership.
\end{proof}

\begin{proof}[Proof of Theorem \ref{th:asy_contagion}]
Given $h^i_0(t|0)=0$, we have $F_i(t|0)=1-e^{-\int_0^t h^i_0(u|0)du}=0$.
\begin{equation*}
\begin{split}
\VE_{\C}(t) &= \E[Y_i(t;Y_j(1),(0,0))] - \E[Y_i(t;Y_j(0),(0,0))] \\
&= \E[Y_i(t;W_j(1),(0,0))] - \E[Y_i(t;W_j(0),(0,0))] \\
\intertext{by Equation (\ref{eq:replace})}
&= \int_0^\infty \E[Y_i(t;w_j,(0,0))] dF_j(w_j|1) -\int_0^\infty \E[Y_i(t;w_j,(0,0)] dF_j(w_j|0)\\
\intertext{ by Corollary \ref{cor:idnatural}}
&= \int_0^\infty \big\{ F_i(w_j|0) + (1-F_i(w_j|0))\E[Y_i(t)|W_i>w_j,W_j=w_j,\X=(0,0)] \big\} d (F_j(w_j|1)-F_j(w_j|0))\\
\intertext{by Theorem \ref{thm:id}}
&= \int_0^\infty \E[Y_i(t)|W_i>w_j,W_j=w_j,\X=(0,0)] d (F_j(w_j|1)-F_j(w_j|0))\\
\intertext{ by $F_i(t|x_i)=0$ for $\forall t>0$ }
&= \int_0^t \E[Y_i(t)|W_i>w_j,W_j=w_j,\X=(0,0)] d (F_j(w_j|1)-F_j(w_j|0))\\
\intertext{since $\E[Y_i(t)|W_i>w_j,W_j=w_j,\X=(0,0)]=0$ for $w_j>t$}
&=  \int_0^t k(w_j) d (F_j(w_j|1)-F_j(w_j|0))\\
\intertext{ by the definition of k(u) in the proof of Theorem \ref{thm:VEInet}}
&= k(w_j)[F_j(w_j|1)-F_j(w_j|0)]\Big|_0^t - \int_0^t F_j(w_j|1)-F_j(w_j|0) dk(w_j)\\
\intertext{by integration by parts}
\end{split}
\end{equation*}
By the definition of $k(u)$ and $F_j(u|x_j)$, we know $k(t)=0$ and $F_j(0|1)-F_j(0|0)=0$, and thus $k(w_j)[F_j(w_j|1)-F_j(w_j|0)]\Big|_0^t=0$. If $\SE(t,w_j,x)>0$, we have $F_j(w_j|1)-F_j(w_j|0)>0$. If $\CE(t,u,u', (0,0))>0$ for $0 \le u <u'<t$, $dk(u)<0$ as shown in the the proof of Theorem \ref{thm:VEInet}. Thus, we have the following conclusions.

When $\SE(t,w_j,x)>0$, $\VE_{\C}(t)$ has the opposite sign as $\CE(t,u',u, (0,0))$.

If $\SE(t,w_j,x)=0$ and $\CE(t,u,u', (0,0))>0$, we have $\VE_{\C}(t)=0$.
\end{proof}

\end{document}